\documentclass[10pt,a4paper]{article}

\usepackage[utf8]{inputenc}
\usepackage[T1]{fontenc}
\usepackage[english]{babel}
\usepackage{lmodern}
\usepackage{microtype}

\usepackage[margin=1in]{geometry}
\usepackage{hyperref}
\usepackage{amsmath, amsfonts, amssymb}
\usepackage{amsthm}
\usepackage{graphicx}
\usepackage{booktabs}
\usepackage{array}
\usepackage{braket}
\usepackage{listings}
\usepackage{xcolor}
\usepackage{enumitem}
\usepackage{caption}
\usepackage{subcaption}
\usepackage{multirow}
\usepackage{cite}
\usepackage[section]{placeins}

\usepackage[flushmargin]{footmisc}
\usepackage[section]{placeins}
\usepackage{float}    

\hypersetup{colorlinks=true, linkcolor=blue!50!black, urlcolor=blue!50!black,
            citecolor=blue!50!black}

\theoremstyle{plain}
\newtheorem{theorem}{Theorem}
\newtheorem{proposition}{Proposition}

\theoremstyle{definition}
\newtheorem{definition}{Definition}

\newtheorem{remark}{Remark}

\newcommand{\diamondnorm}[1]{\left\lVert #1 \right\rVert_\diamond}
\newcommand{\opnorm}[1]{\left\lVert #1 \right\rVert_\infty}
\newcommand{\trnorm}[1]{\left\lVert #1 \right\rVert_1}
\newcommand{\cE}{\mathcal{E}}
\newcommand{\cO}{\mathcal{O}}
\newcommand{\cH}{\mathcal{H}}
\newcommand{\cD}{\mathcal{D}}
\newcommand{\cS}{\mathcal{S}}
\newcommand{\cN}{\mathcal{N}}
\newcommand{\spec}{\mathrm{spec}}
\newcommand{\eps}{\varepsilon}

\definecolor{codebg}{rgb}{0.97,0.97,0.97}
\definecolor{codecomment}{rgb}{0.4,0.5,0.4}
\definecolor{codekw}{rgb}{0.0,0.0,0.6}
\lstdefinestyle{qcivet}{
  backgroundcolor=\color{codebg},
  commentstyle=\color{codecomment}\itshape,
  keywordstyle=\color{codekw}\bfseries,
  basicstyle=\ttfamily\footnotesize,
  breaklines=true,
  captionpos=b,
  keepspaces=true,
  showstringspaces=false,
  numbers=none,
  language=Python,
  frame=single,
  framesep=4pt,
  rulecolor=\color{black!20},
}
\title{QCIVET: A Quantum--Classical Pipeline Integrity\\
Framework with Contract-Based Subtype Verification\\
and Hash-Chained Audit Traces}
\author{%
  Esra Yeniaras\thanks{Corresponding author. Quantum Security and Post-Quantum Cryptography Researcher; previously Assistant Professor in Cyber Security at Copenhagen School of Business (EK -- Erhvervsakademi K{\o}benhavn). Email: esramath@gmail.com}
  \and
  Muhammad Amin Karimov\thanks{Department of Computer Engineering, Istanbul Gelisim University. Email: muhammad.amin@ogr.gelisim.edu.tr}%
}

\date{}
\usepackage{tikz}
\usetikzlibrary{positioning,shapes,arrows.meta,calc}
\usepackage{forest}
\usepackage{xcolor}
\begin{document}
\maketitle

\begin{abstract}
Hybrid quantum--classical pipelines now sit behind decisions
that range from drug-binding estimates to real-time fraud
alerts and customer-side auditing of cloud quantum processing
unit (QPU) services, but the integrity tooling we audit them
with is mostly borrowed from the classical world and therefore
blind to the quantum stages in the middle. We propose QCIVET,
a contract-based integrity-verification framework that treats
a hybrid pipeline as a sequence of stages, each carrying an
explicit specification, and audits the whole sequence in two
ways at once. Syntactic integrity is enforced by a
hash-chained audit trail with optional external anchoring;
semantic integrity is enforced at the quantum stages by a
calibrated observable-deviation test rooted in the
behavioural-subtyping discipline of Liskov and Wing. We prove
that the test is sound under the diamond-norm distance
between channels, conditionally complete when the observable
family is informationally complete (with an explicit constant
$C(\cO_A) = 2\sqrt{2}$ for single-qubit Pauli families), and
compositional under inheritance chains. We also single out a
class of ``Z-only-sneaky'' overrides that pass weak,
single-Pauli contracts but are immediately exposed by
multi-Pauli ones; this is a design lesson that follows from
our proofs. The framework is validated under the
calibration-derived noise models of two production IBM Quantum
processors (Eagle r3 and Heron r2), and the subtype-separation
protocol is run end-to-end on a real \texttt{ibm\_fez}
(Heron r2) processor, where the predicted sneaky fingerprint
survives intact ($1.401 \to 1.386 \to 1.420$ along the ideal
$\to$ simulated $\to$ real chain). The framework is
instantiated end-to-end on three pipelines: variational
quantum eigensolver (VQE) for early-stage drug discovery,
quantum-assisted fraud detection, and customer-side auditing
of a cloud QPU service. The reference implementation,
including a real-time engine whose median per-stage commit
latency is below 0.1 ms, is open source.

\end{abstract}

\section{Introduction}
\label{sec:intro}

Quantum computing is no longer a thing one demonstrates in a
single notebook. Pharmaceutical companies use the variational
quantum eigensolver (VQE) for early-stage drug-binding studies;
financial institutions experiment with quantum kernels for
fraud detection and quantum amplitude estimation for
derivatives risk; cloud providers sell access to processors
exceeding one hundred qubits. In every one of these settings
the quantum subroutine is embedded in a longer hybrid pipeline,
and \emph{the result of the quantum stage feeds a classical
decision that may be hard to reverse}. A miscalibrated VQE
energy can mislead a clinical trial; a poisoned kernel matrix
can suppress a fraud alert; a silently re-routed cloud job can
return results from a backend the customer never agreed to use.

The trouble is that the integrity tooling deployed in classical
software supply chains was never designed for these
pipelines. Frameworks such as in-toto~\cite{intoto}, the Supply-chain
Levels for Software Artifacts (SLSA)~\cite{slsa}, and
Sigstore~\cite{sigstore} attest that an artefact \emph{was
built by a given recipe}: they bind a cryptographic hash to a
declared build process. They do not
attest that the artefact, when run on a noisy quantum
processor, \emph{preserved a stated semantic contract} under
realistic device noise. Recent ``quantum-resistant integrity''
work~\cite{ai_supply_chain_pqc, quantum_resistant_trusted}
re-engineers the cryptographic primitives so that classical
attestations stay forge-resistant against future quantum
adversaries; this is orthogonal, since the workflow itself is
still treated as classical.

\paragraph{The gap.}None of the existing tools resolve the three scenarios above
(a miscalibrated VQE energy, a poisoned kernel matrix, a
silently rerouted cloud job) within a single framework.
Classical supply-chain tools such as in-toto, SLSA, and
Sigstore answer one question well: ``was this artefact built
from the declared recipe?''. They take a hash of the build
process and confirm that nobody quietly swapped a step or a
parameter along the way. This is structural integrity: the
audit log itself has not been tampered with. What these tools
cannot answer is a different and equally important question:
``did the quantum stage actually behave the way it was
supposed to behave on the hardware that ran it?''. A circuit
can be transpiled and submitted exactly as recorded, with
every hash matching, and still produce a result that is too
far from the contract the developer wrote down, because the
qubits drifted, or because a sneaky variant of the circuit was
substituted that happens to look correct in one specific
measurement direction but misbehaves everywhere else. This is
behavioural integrity: the run actually delivered what the
contract promised. The reverse failure mode is just as real.
A pipeline whose quantum stage was well-behaved is also broken
if the audit log was tampered with after the fact, since later
auditors no longer know what was actually run. Quantum-resistant
signing schemes harden the structural side against future
cryptographic attacks but leave the behavioural question
untouched. What is missing, and what QCIVET provides, is a
single discipline that audits both kinds of integrity at once:
a hash-chained audit trail that survives tampering plus an
observable-deviation check that the quantum output sits inside
a calibrated tolerance of the contract a developer specified,
with a formal soundness and completeness link between the two.

\paragraph{Contributions.} This paper makes the following
contributions.
\begin{enumerate}
\item We identify and formalise the integrity-verification
problem for hybrid quantum--classical pipelines as distinct
from both classical supply-chain integrity and quantum-resistant
signing of classical pipelines (Section~\ref{sec:related}).
\item We develop QCIVET: a contract-based subtyping discipline
in which each stage carries an explicit spec, a hash-chained
audit trail in the spirit of secure audit
logs~\cite{schneier_kelsey} and Merkle commitments~\cite{merkle},
an optional external anchor for global-rewrite detection, and
a real-time engine that streams stage commits
(Sections~\ref{sec:bg} and~\ref{sec:framework}).
\item We prove three formal properties: \emph{soundness}
(channel-level closeness implies observable-level closeness
via H{\"o}lder's inequality), \emph{conditional completeness}
(observable-level closeness implies channel-level closeness
when the observable family is informationally complete, with
$C(\cO_A) = 2\sqrt{2}$ for the single-qubit Pauli family), and
\emph{compositionality} along inheritance chains. We also
characterise the family of ``sneaky overrides'' that pass
weak contracts but fail informationally complete ones
(Section~\ref{sec:formal}).
\item We validate the framework under realistic device noise
via the calibrated noise models of two production IBM Quantum
processors, and confirm the separation of valid, invalid, and
sneaky subtypes on a real \texttt{ibm\_fez} (Heron r2)
processor, where the predicted sneaky fingerprint survives
intact ($1.401 \to 1.386 \to 1.420$ along the ideal
$\to$ simulated $\to$ real chain) (Section~\ref{sec:exp}).
\item We instantiate the framework end-to-end in three domains:
VQE-driven drug discovery, quantum-assisted fraud detection,
and cloud QPU auditing (Section~\ref{sec:rt}). For each, we
exhibit a six-stage pipeline, a threat model, and four scenarios
(clean, local tampering, semantic drift, global rewrite); the
engine catches every attack at the appropriate scope.
\end{enumerate}

\paragraph{Software availability and reproducibility.}
The reference implementation, the simulation scripts, the
device-noise validation, and the three end-to-end demos are
released as open-source software under the MIT license at the
project repository
\url{https://github.com/schrodinket/QCIVET}.\footnote{\label{fn:repo}%
\url{https://github.com/schrodinket/QCIVET}}
The repository contains the verification engine
(\texttt{qcivet\_realtime.py}), the experimental scripts that
reproduce every figure and table in this paper, a real-hardware
validation script (\texttt{quantum\_oop\_\allowbreak real\_qpu.py}) that
runs the subtype-separation protocol on an IBM quantum
processor accessed through the IBM Quantum cloud, three
application demonstrators (VQE, fraud detection, cloud QPU
auditing), and a hash-chain attack-scenario prototype. All
simulation-based experiments run on commodity hardware in a
few minutes; the device-noise experiments additionally require
a free IBM Quantum account, and the real-hardware validation
requires submitting jobs to the IBM Quantum cloud. Each figure
and table caption in the body of the paper points to the
specific script (and, where applicable, function) that
generates it.

\paragraph{Organization of this paper.} 
Section~\ref{sec:related} surveys the related literature and
positions QCIVET. Section~\ref{sec:bg} recalls the contract
and behavioural-subtyping vocabulary and the quantum-channel
formalism used throughout. Section~\ref{sec:framework}
introduces the QCIVET subtyping discipline, and
Section~\ref{sec:formal} proves its soundness, conditional
completeness, and compositionality, including a formal
characterisation of sneaky overrides. Section~\ref{sec:exp}
reports the experimental validation under both simulated and
real IBM hardware (\texttt{ibm\_fez}). Section~\ref{sec:rt}
presents the real-time engine and instantiates the framework
on three application domains: VQE-driven drug discovery,
quantum-assisted fraud detection, and customer-side cloud QPU
auditing. Section~\ref{sec:threat} fixes the threat model and
discusses security guarantees. Section~\ref{sec:conc}
concludes and outlines directions for future work.

\section{Related Work and Positioning}
\label{sec:related}

QCIVET sits at the intersection of four research threads:
classical software supply-chain integrity, post-quantum
cryptographic protections of classical workflows,
quantum-software contracts and refinement, and trust
mechanisms for cloud quantum platforms. We surveyed each
thread systematically, including a structured literature
review using a 200M-paper academic search engine to confirm
the absence of prior unified frameworks (see
Section~\ref{sec:related:summary}). To our knowledge, no
prior work integrates hash-chain integrity, behavioural
observable contracts, behavioural subtyping for quantum
channels, and the multi-stage hybrid pipeline shape into a
single framework with formal soundness, completeness, and
compositionality guarantees.

\subsection{Classical Software Supply-Chain Integrity}
\label{sec:related:classical}

The state of the art for classical pipelines is built around
attested provenance. in-toto~\cite{intoto} formalises a
``farm-to-table'' chain of attestations linking sources to
artefacts; SLSA~\cite{slsa} layers a
maturity model on top, with active analysis of adoption
challenges~\cite{patel_slsa_2025}; Sigstore~\cite{sigstore} offers
keyless signing and a public transparency log. Recent
work~\cite{hassanshahi_supply_chain} systematises provenance,
tamper resistance, and build integrity practices across this
ecosystem. Closely related, hash-chain-based integrity
verification has been used in adjacent
domains~\cite{kim_blockchain_hash_chain,merkle,schneier_kelsey},
notably blockchain transaction auditing, and QCIVET inherits
the chain-of-commitments idea from this line.

These tools assume that the artefact under attestation is a
binary, image, or package whose integrity is fully captured by
a cryptographic hash of its bytes. This assumption holds for
classical software but breaks for the quantum stages of a
hybrid pipeline in two distinct ways. First, hardware noise:
the same quantum circuit, when executed on a real QPU, can
produce different outputs depending on the calibration drift
and the prevailing single- and two-qubit error rates. Second,
semantic substitution: a byte-level-conformant ``sneaky''
circuit (formally characterised in
Section~\ref{sec:formal}) may pass a weak observable contract
while violating an informationally complete one. Two artefacts
with identical hashes can therefore exhibit very different
quantum behaviours. QCIVET inherits the hash-chain and
external-anchor pattern from this literature (Sigstore Rekor
and RFC 3161~\cite{rfc3161} timestamping authorities are
concrete instantiations of our \texttt{ExternalAnchor}
interface) and extends it with an observable-deviation check
at each quantum stage. Table~\ref{tab:rw:classical} contrasts
these classical supply-chain tools with QCIVET on three
dimensions: hash-chain integrity, external anchoring, and
quantum-stage scope.

\begin{table}[h]
\centering
\footnotesize
\caption{Classical supply-chain integrity versus QCIVET.}
\label{tab:rw:classical}
\begin{tabular}{|p{0.30\linewidth}|p{0.30\linewidth}|p{0.30\linewidth}|}
\hline
\textbf{Classical tools} & \textbf{What they verify} & \textbf{QCIVET adds} \\
in-toto~\cite{intoto}, SLSA~\cite{slsa}, Sigstore~\cite{sigstore} & & (this work) \\
\hline
Hash chain of build steps     & Bytes match the recipe (no quantum awareness) & Same hash chain, kept for hybrid quantum-classical pipelines \\
\hline
External anchor (Sigstore Rekor, RFC 3161~\cite{rfc3161}) & Log not rewritten (classical artefacts only) & Same anchor, kept \\
\hline
No quantum stage in scope     & Assumes bytes equal $=$ behaviour equal       & Observable check at every quantum stage $\star$ \\
\hline
\end{tabular}
\end{table}
\subsection{Quantum-Resistant Integrity for Classical Workflows}
\label{sec:related:pqc}

A growing body of work studies how classical signing and
attestation should evolve when quantum adversaries become
practical \cite{NISTPQCProcess,NISTAdditionalPQCSignatures, Yeniaras2026,YeniarasFaster,YeniarasImproved}. MBOM-PQC (Model Bill of Materials with Post-Quantum
Cryptography)~\cite{ai_supply_chain_pqc} proposes
post-quantum-safe signing for AI/ML model lineage;
\cite{quantum_resistant_trusted} integrates ML-DSA
(Module-Lattice-Based Digital Signature Algorithm) into a
Trusted Platform Module (TPM)-based remote-attestation flow;
the broader transition to NIST post-quantum cryptography (PQC)
algorithms~\cite{nist2024fips203,nist2024fips204,nist2024fips205}
touches all of these. None of them changes the assumption that
the pipeline itself is classical: the cryptographic primitives
change, the workflow they protect does not. QCIVET is
orthogonal: we keep classical SHA-256 in our chain (a deployer
is free to substitute SHA-3 or a PQC-safe message authentication
code (MAC)) and instead
change the pipeline shape, admitting quantum stages and giving
them a semantic contract. The two approaches compose readily;
a deployment could combine them by signing each anchor
commitment with an ML-DSA signature. Table~\ref{tab:rw:pqc}
situates QCIVET against classical (RSA, ECDSA) and PQ-hardened
(MBOM-PQC, ML-DSA) signing approaches.

\begin{table}[h]
\centering
\footnotesize
\caption{Quantum-resistant signing versus QCIVET.}
\label{tab:rw:pqc}
\begin{tabular}{|p{0.30\linewidth}|p{0.30\linewidth}|p{0.30\linewidth}|}
\hline
\textbf{Classical approach} & \textbf{PQ-hardened approach} & \textbf{QCIVET} \\
(RSA, ECDSA) & MBOM-PQC~\cite{ai_supply_chain_pqc}, ML-DSA + TPM~\cite{quantum_resistant_trusted} & (this work) \\
\hline
RSA signature on classical pipeline   & ML-DSA signature on classical pipeline      & SHA-256 (kept) on hybrid quantum-classical pipeline \\
\hline
No quantum stages                     & No quantum stages                            & Quantum stages with semantic contracts $\star$ \\
\hline
Byte-hash only                        & Byte-hash only                               & Byte-hash $+$ quantum observable-deviation check $\star$ \\
\hline
\end{tabular}

\vspace{0.5em}
\begin{minipage}{0.95\linewidth}
\scriptsize
\noindent\textit{Note.} ML-DSA is the lattice-based signature standardised by NIST FIPS~204~\cite{nist2024fips204}; the wider NIST post-quantum suite also includes ML-KEM (FIPS~203,~\cite{nist2024fips203}) for key encapsulation and SLH-DSA (FIPS~205,~\cite{nist2024fips205}) as a hash-based signature alternative. Any of these can be substituted for ML-DSA in a PQ-hardened pipeline; QCIVET is orthogonal to that choice.
\end{minipage}
\end{table}

\subsection{Quantum Cryptography and Quantum-Augmented Integrity}
\label{sec:related:qcrypto}

A parallel stream uses quantum mechanics as a cryptographic
primitive: BB84-derived hybrid encryption~\cite{mozo_bb84_aes},
where BB84 is a quantum key distribution (QKD) protocol;
quantum hash functions (QHFs) and quantum-walk-based
primitives~\cite{wang_qhf_blockchain}, with hybrid hash
frameworks targeting post-quantum security~\cite{bhadane2025hybrid}. Here the
quantum component is the security mechanism; the workflow it
protects is once again classical. QCIVET inverts the picture:
the workflow is quantum, the security mechanism is classical.
These approaches are composable: a deployer may pair QCIVET's 
hash-chain anchors with a quantum-walk hash function, or wrap 
each anchor commitment with a PQC signature for non-repudiation.
Table~\ref{tab:rw:qcrypto} highlights this contrast: QCIVET
treats the quantum component as the workload to be audited
rather than as a cryptographic primitive.

\begin{table}[h]
\centering
\footnotesize
\caption{Quantum cryptography and quantum-augmented integrity versus QCIVET.}
\label{tab:rw:qcrypto}
\begin{tabular}{|p{0.30\linewidth}|p{0.30\linewidth}|p{0.30\linewidth}|}
\hline
                  & \textbf{Quantum-augmented integrity} & \textbf{QCIVET} \\
                  & BB84 / QKD~\cite{mozo_bb84_aes}, quantum hash~\cite{wang_qhf_blockchain} & (this work) \\
\hline
Workflow type     & Classical pipeline & Hybrid quantum-classical pipeline $\star$ \\
\hline
Security mechanism & Quantum (BB84 / QKD, quantum hash) & Classical (SHA-256, hash chain) \\
\hline
Quantum used for  & Building security primitives & Computation that must be audited $\star$ \\
\hline
Composability     & --- & Composes with PQC signatures and quantum-walk hashes \\
\hline
\end{tabular}
\end{table}
\FloatBarrier

\subsection{Quantum Software Contracts and Refinement}
\label{sec:related:qse}

A small but growing literature treats correctness of quantum
software at the specification level. Three families are
relevant: design-by-contract for individual circuits or
modules, Hoare logic and refinement calculi for whole
programs, and compiler verification. Figure~\ref{fig:rw_landscape}
summarises the landscape and indicates where QCIVET sits
relative to it. The three established families operate at the
granularity of a single circuit, a single program, or a
single compilation pass, and they target either functional
correctness or program development. QCIVET addresses a
different layer: runtime integrity of a multi-stage hybrid
quantum-classical pipeline against an adversary, using
operationally measurable Pauli observables and a hash-chained
audit trail. We expand on each family in the following
paragraphs.

\paragraph{Design-by-contract for quantum software.}
The closest prior work is the design-by-contract framework of
Yamaguchi and Yoshioka~\cite{yamaguchi_dbc_quantum}, a
Python-embedded language that lets a programmer attach
pre/post-state assertions and assertions over the statistical
processing of measurement results to individual quantum
circuits. ScaffML~\cite{jin_zhao_scaffml} provides analogous
pre- and post-conditions at the Scaffold module level. QCIVET
differs from both along three substantive axes: (i) we operate
at the \emph{pipeline} granularity, treating each quantum
stage as one node in a hash-chained multi-stage workflow,
rather than asserting on a single circuit or module; (ii) we
provide a behavioural-subtyping foundation with the
sneaky-subtype impossibility result
(Proposition~\ref{prop:sneaky}) which is absent from these
state- and module-equality frameworks; (iii) we integrate
observable contracts with cryptographic hash-chain integrity,
giving simultaneous coverage of in-flight semantic drift and
post-hoc audit-trail tampering, neither of which intra-circuit
or intra-module assertions address.

\begin{figure}[h]
\centering
\scriptsize
\begin{forest}
  for tree={
    grow=east,
    parent anchor=east,
    child anchor=west,
    edge={draw=gray!60, thick},
    s sep=1pt,
    l sep=8pt,
    inner sep=4pt,
    rounded corners=2pt,
    draw=blue!40,
    fill=blue!5,
    font=\sffamily\scriptsize,
    anchor=west,
    align=left,
  }
  [Quantum Software Contracts, draw=blue!70, fill=blue!15, font=\sffamily\scriptsize\bfseries
    [Yamaguchi (single circuit), fill=cyan!8, draw=cyan!50]
    [ScaffML (single module), fill=cyan!8, draw=cyan!50]
    [Quantum Hoare Logic (theory), draw=teal!60, fill=teal!10, font=\sffamily\scriptsize\bfseries
      [aQHL (projections), fill=teal!5, draw=teal!40]
      [qRHL (relational), fill=teal!5, draw=teal!40]
      [Feng-Zhou (refinement orders), fill=teal!5, draw=teal!40]
    ]
    [Compiler / Type theory, draw=violet!50, fill=violet!10, font=\sffamily\scriptsize\bfseries
      [CertiQ (compiler verification), fill=violet!5, draw=violet!30]
      [FJQuantum (type theory), fill=violet!5, draw=violet!30]
      [Process algebras, fill=violet!5, draw=violet!30]
    ]
  ]
\end{forest}

\vspace{0.4em}
\centerline{\rule{0.4\linewidth}{0.3pt}}
\vspace{0.4em}

\begin{forest}
  for tree={
    grow=east,
    parent anchor=east,
    child anchor=west,
    edge={draw=orange!60, thick},
    s sep=1pt,
    l sep=8pt,
    inner sep=4pt,
    rounded corners=2pt,
    draw=orange!50,
    fill=orange!8,
    font=\sffamily\scriptsize,
    anchor=west,
    align=left,
  }
  [{\bfseries QCIVET}: Pipeline integrity (multi-stage hybrid + adversary), draw=orange!80, fill=orange!20, font=\sffamily\scriptsize\bfseries, line width=0.8pt
    [Pauli observables (operationally measurable)]
    [Behavioural subtyping (Liskov-Wing)]
    [Sneaky impossibility (Proposition~\ref{prop:sneaky})]
    [Hash-chain integrity (cryptographic)]
  ]
\end{forest}
\caption{Landscape of quantum software contract frameworks
and where QCIVET sits. The three families above (blue/teal/violet)
operate at the level of a single circuit, single program, or
compile time. QCIVET (orange) targets a different layer:
runtime integrity of a multi-stage hybrid quantum-classical
pipeline under an adversarial threat model.}
\label{fig:rw_landscape}
\end{figure}
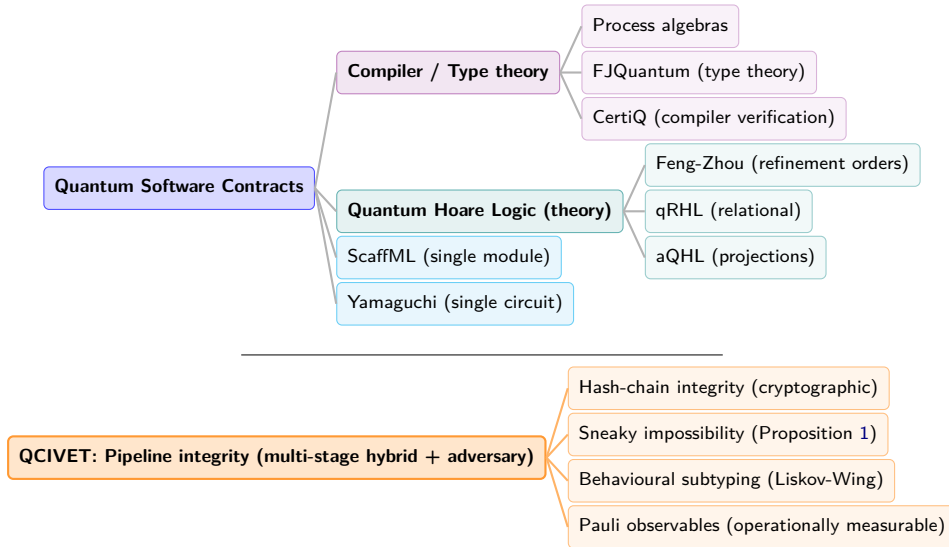
\FloatBarrier

\paragraph{Quantum Hoare logic and refinement calculi.}
The semantic backbone for observable specifications is quantum
Hoare logic (qHL), in which observables play the role of
predicates that are pushed back through the program by a
weakest-precondition transformer~\cite{ying_quantum_hoare}. Applied quantum Hoare
logic (aQHL)~\cite{aqhl} restricts pre- and post-conditions to
projections, simplifying verification while admitting rules for
robustness and output error bounds; quantum relational Hoare
logic (qRHL)~\cite{qrhl} formulates relational invariants
between two programs. Most directly comparable to our
substitutability test, Feng and
Zhou~\cite{feng_zhou_refinement_orders} (and the related
TOSEM calculus of
Feng-Zhou-Xu-Xu~\cite{feng_zhou_subtyping}) provide the first
comprehensive study of refinement orders for quantum programs
under projector-based, effect-based, and set-of-effects-based
specifications, characterising when one quantum program can
replace another. Their semantic, language-independent
treatment, together with the order-theoretic characterisations
in terms of complete positivity and the Smyth/Hoare orders,
establishes the foundational backbone for substitutability in
the quantum setting, and is the natural denotational
counterpart to our operational, runtime-measurable contract.
QCIVET extends this foundation along three complementary axes:
(i) where Feng-Zhou focus on \emph{program development}
(stepwise refinement of a specification toward an
implementation), QCIVET addresses the runtime integrity of an
already-implemented hybrid pipeline against an adversary;
(ii) where Feng-Zhou take projectors and effects as predicates,
which is the right level of generality for development-time
reasoning, QCIVET projects onto operationally measurable Pauli
expectation values with a calibrated tolerance budget, which a
deployer can compute directly from a finite shot count on real
hardware; (iii) where Feng-Zhou prove order-theoretic
correspondences, QCIVET proves runtime soundness and
completeness up to a noise floor, integrated with a hash-chain
audit trail. The two frameworks are intentionally
complementary: the refinement orders
of~\cite{feng_zhou_refinement_orders,feng_zhou_subtyping}
provide the denotational semantics that underwrites our
subtype obligations, while QCIVET supplies the operational,
hardware-validated projection of that semantics needed for
runtime audit on noisy hybrid pipelines. We see our work as
the audit-time companion to the refinement-time framework of
Feng-Zhou. We also note Li et al.'s
projection-based runtime assertions~\cite{li_proq}, an early
runtime-monitoring approach for quantum programs that we
extend by tying assertions to a hash-chain audit trail and a
hybrid pipeline.

\paragraph{Quantum compiler verification and language
foundations.}
CertiQ~\cite{shi_certiq} verifies properties of a realistic
quantum compiler with a contract-based methodology. While
related in spirit, the compiler-verification problem is
disjoint from pipeline integrity: CertiQ guarantees that the
compiler produces semantically equivalent circuits, whereas
QCIVET guarantees that runtime executions of a (possibly
adversarially modified) compiled pipeline satisfy a calibrated
contract. Object-calculus work such as FJQuantum (a Featherweight Java
calculus extended with quantum primitives)~\cite{fjquantum}
and quantum process algebras~\cite{ying_quantum_processes}
address typing and equivalence and are complementary to our
runtime concerns.
General quantum software engineering surveys map the broader
landscape~\cite{dwivedi_qse_survey, qse_lifecycle}.
Table~\ref{tab:rw:qse} summarises how QCIVET differs from the
closest quantum software contract frameworks
(Yamaguchi-Yoshioka, ScaffML, Feng-Zhou).

\begin{table}[h]
\centering
\footnotesize
\caption{Quantum software contract frameworks versus QCIVET.}
\label{tab:rw:qse}
\begin{tabular}{|p{0.30\linewidth}|p{0.30\linewidth}|p{0.30\linewidth}|}
\hline
\textbf{Yamaguchi-Yoshioka, ScaffML} & \textbf{Feng-Zhou refinement} & \textbf{QCIVET} \\
\cite{yamaguchi_dbc_quantum,jin_zhao_scaffml} (design-by-contract) & \cite{feng_zhou_refinement_orders,feng_zhou_subtyping} (refinement orders) & (this work) \\
\hline
Single quantum circuit or module & Single quantum program & Multi-stage hybrid quantum-classical pipeline $\star$ \\
\hline
State or measurement assertions & Projector or effect predicates & Observable-deviation contract on quantum output $\star$ \\
\hline
No hash-chain audit trail & No hash-chain audit trail & Hash-chain audit trail with external anchor \\
\hline
No sneaky-subtype result & Refinement, not adversarial & Sneaky-subtype impossibility result for quantum stages $\star$ \\
\hline
\end{tabular}
\end{table}
\FloatBarrier

\subsection{Cloud QPU Trust}
\label{sec:related:cloud}

Customer-side trust in cloud quantum services has emerged as
an active subfield. Four complementary approaches deserve
explicit comparison with the cloud-auditing demonstrator of
Section~\ref{sec:apps:cloud}.

\paragraph{Device fingerprinting.}
Wu et al.~\cite{wu_dynamic_fingerprinting,
wu_dynamic_fingerprinting_extended} authenticate a quantum
device by probing it with calibration circuits and comparing
the device-side error fingerprint against a user-side
expectation, detecting machine substitution and profile
fabrication attacks. The approach is empirically validated on
seven IBM devices. Fingerprinting authenticates the
\emph{device}; QCIVET authenticates the \emph{result}. A device
might pass the fingerprint test (correct hardware) yet still
return a result outside the customer's contract, for
example, because the upstream specification was tampered with,
because middleware re-routed circuits, or because the
calibration drifted between the fingerprint and the workload.
The two are complementary and could be deployed together.

\begin{table}[h]
\centering
\scriptsize
\caption{Cloud QPU trust mechanisms versus QCIVET.}
\label{tab:rw:cloud}
\begin{tabular}{|p{0.18\linewidth}|p{0.16\linewidth}|p{0.16\linewidth}|p{0.16\linewidth}|p{0.20\linewidth}|}
\hline
                  & \textbf{Device fingerprinting} & \textbf{Distributed shots} & \textbf{Quantum PUF} & \textbf{QCIVET} \\
                  & Wu et al.~\cite{wu_dynamic_fingerprinting,wu_dynamic_fingerprinting_extended} & Upadhyay-Ghosh~\cite{upadhyay_ghosh_2022,upadhyay_ghosh_2024} & Phalak et al.~\cite{phalak_quantum_puf} & (this work) \\
\hline
Authenticates     & The quantum device & Hardware via voting & Hardware identity & The quantum result against a contract $\star$ \\
\hline
Devices needed    & One quantum device & Multiple quantum devices & One quantum device & One quantum device \\
\hline
Detects spec or middleware tampering & No & No & No & Yes $\star$ \\
\hline
Composes with QCIVET & Yes & Yes & Yes & --- \\
\hline
\end{tabular}

\vspace{0.5em}
\begin{minipage}{0.95\linewidth}
\scriptsize
\noindent\textit{Note.} A complementary cryptographic-protocol line (Leichtle et al.~\cite{leichtle_verification}) provides composable verification with statistical guarantees, but typically requires trap-qubit constructions on measurement-based graphs; QCIVET targets the same threat surface with a single device and a tracer observable.
\end{minipage}
\end{table}
\FloatBarrier

\paragraph{Distributed shot allocation.}
Upadhyay and Ghosh~\cite{upadhyay_ghosh_2022,
upadhyay_ghosh_2024} distribute repeated executions across
multiple hardware options and use majority voting (or adaptive
allocation) to detect tampered hardware. This requires
multiple physical devices and is most effective for pure
quantum workloads. QCIVET works on a single device by checking
against a calibrated observable tolerance, requires no
replication, and detects tampering not only of hardware but
also of specifications and pipeline middleware. The four
mechanisms discussed in this section are summarised in
Table~\ref{tab:rw:cloud}.

\paragraph{Quantum physically unclonable functions.}
Phalak et al.~\cite{phalak_quantum_puf} propose quantum
physically unclonable functions (PUFs)
for hardware authentication, achieving strong inter-device
Hamming-distance separation. PUFs authenticate the
\emph{identity} of a hardware unit; QCIVET authenticates the
\emph{integrity of a result} against a contract. The two
operate at orthogonal layers and compose without conflict.

\paragraph{Cryptographic delegation protocols.}
A cryptographic-protocol line, exemplified by Leichtle et
al.~\cite{leichtle_verification}, develops composable
verification protocols that interleave computation rounds with
test rounds for malicious-behaviour detection. Such protocols
offer formal, statistical guarantees against arbitrary malicious
servers but typically require trap-qubit constructions on
measurement-based quantum-computation graphs and incur
significant per-circuit overhead. QCIVET targets the same threat
surface at a lower formal level but with operational practicality:
the customer needs only to choose a tracer observable and a
tolerance budget, with no extra cryptographic machinery on
the quantum side.

\subsection{Hybrid Quantum-Classical Applications and Audit Trails}
\label{sec:related:apps}

The three application domains in Section~\ref{sec:apps}
intersect with separate application-specific literatures.

\paragraph{VQE for chemistry and materials.}
The variational quantum eigensolver has been studied
extensively as an algorithm
~\cite{fedorov_vqe_review,tilly_vqe_review,harville_vqe_recent},
with attention to ansatz design, optimiser robustness, and
device-noise mitigation. To our knowledge, no prior work
proposes integrity or audit-trail mechanisms specifically
designed for VQE workflows in pharmaceutical
applications, even though general pharmaceutical audit-trail
standards exist independently~\cite{saxena_pharma_audit,fda_21cfr11}. Our
VQE drug-discovery demonstrator (Section~\ref{sec:apps:vqe})
is the first such mechanism we are aware of.

\paragraph{Hybrid quantum-classical financial systems.}
Concurrent work on HQFS (Hybrid Quantum-Classical Financial
System)~\cite{nayak_hqfs} proposes a hybrid
quantum-classical financial system that combines VQE-based
forecasting with post-quantum cryptographic signing of
allocation records and an audit trail linking decisions to
model state and inputs. QCIVET's fraud-detection demonstrator
(Section~\ref{sec:apps:fraud}) addresses an adjacent setting
but with three orthogonal differences: (i) HQFS uses
post-quantum \emph{signatures} to authenticate records after
the fact, while QCIVET uses a hash-chain \emph{integrity
verifier} that detects tampering at commit time \emph{and}
via post-pipeline replay, including the case in which an
attacker has obtained the signing keys; (ii) HQFS provides no
semantic-level check on the quantum stage's output, whereas
QCIVET enforces a calibrated observable contract; (iii) HQFS
does not address a behavioural-subtyping discipline, leaving
the system blind to the sneaky-subtype attack pattern of
Proposition~\ref{prop:sneaky}. The two approaches are
complementary: a deployer could combine HQFS-style
post-quantum signatures (for non-repudiation) with QCIVET's
observable contracts and hash-chain integrity (for tamper
detection and semantic checks).
\begin{table}[H]
\centering
\footnotesize
\caption{HQFS hybrid quantum-classical financial system versus QCIVET fraud-detection demonstrator.}
\label{tab:rw:apps}
\begin{tabular}{|p{0.30\linewidth}|p{0.30\linewidth}|p{0.30\linewidth}|}
\hline
                  & \textbf{HQFS} & \textbf{QCIVET} \\
                  & Nayak et al.~\cite{nayak_hqfs} & (fraud-detection demo) \\
\hline
Tamper detection  & Post-quantum signatures after the fact & Hash-chain integrity at commit time and on replay $\star$ \\
\hline
Quantum-stage semantic check & None (signature only) & Calibrated observable contract on quantum output $\star$ \\
\hline
Sneaky-subtype coverage & None & Behavioural-subtyping discipline detects sneaky overrides $\star$ \\
\hline
Composability & --- & Composes with HQFS-style post-quantum signatures \\
\hline
\end{tabular}
\end{table}

\paragraph{Hybrid security via quantum machine learning.}
A separate strand uses hybrid quantum-classical machine
learning (QML) for security tasks such as threat
detection~\cite{evans_hybrid_threat,passo_hybrid_security,
akter_quantum_ml_security}. This direction treats the quantum
component as a tool inside a security application, whereas
QCIVET treats the quantum component as an asset to be
\emph{protected}: the pipeline that contains the QML
classifier is what we audit. Surveys of quantum software
engineering, including requirements
engineering~\cite{sepulveda_quantum_se}, catalogue practices
for hybrid systems but do not address pipeline integrity
specifically. Table~\ref{tab:rw:apps} contrasts HQFS with
QCIVET's fraud-detection demonstrator, highlighting the
orthogonal choices on tamper detection, semantic checks, and
sneaky-subtype coverage.

\subsection{Behavioural Subtyping in Classical Software}
\label{sec:related:subtyping}

QCIVET's behavioural-subtyping discipline traces back to
Liskov and Wing's classical
formulation~\cite{liskov_wing} and Meyer's
design-by-contract~\cite{meyer_oosc}; subsequent work has
extended the idea to product-line settings via
feature-oriented contracts~\cite{thum_feature_oriented_contracts}
and to runtime monitoring. QCIVET adopts a complementary
engineering vocabulary to the substantial recent progress on
quantum substitutability via \emph{refinement} and
\emph{compliance}, notably
Feng-Zhou~\cite{feng_zhou_refinement_orders,feng_zhou_subtyping}:
we cast the same substitutability question in the Liskov-Wing
register of behavioural subtyping, which gives us an intuitive
engineering interpretation (``$B$ is a behavioural subtype of
$A$ if $B$ can replace $A$ without observable surprise to a
contract-respecting client'') and a clean operational reading
of the sneaky-subtype phenomenon
(Proposition~\ref{prop:sneaky}). The two vocabularies are
translatable: every result we prove can be interpreted within the
refinement-order backbone of Feng-Zhou, and conversely, our
observable-deviation contract provides a runtime-measurable
instance of their refinement relation.

\subsection{Summary and Positioning}
\label{sec:related:summary}

To validate the absence of prior unified frameworks we
conducted a systematic search using Elicit AI, a
literature-search tool indexing over 200~million academic
papers. Across five queries spanning the four threads above,
we recorded three explicit ``no prior work'' findings: (i)
no source proposes a hash-chain-based integrity verification
framework for hybrid quantum-classical software pipelines
that also includes observable contracts and behavioural
subtyping; (ii) no source extends classical supply-chain
frameworks (SLSA, in-toto, Sigstore) to hybrid
quantum-classical pipelines with semantic-level checks; (iii)
no source addresses integrity or audit-trail mechanisms
specifically for VQE pharmaceutical workflows. Manual
follow-up across arXiv and IEEE confirmed these gaps and
identified the closest individual contributions cited above.
A complete record of the queries and returned references is
maintained in the project repository (footnote~\ref{fn:repo}).

Table~\ref{tab:positioning} summarises QCIVET's position
relative to the four most adjacent threads.

\begin{table}[H]
\centering
\footnotesize
\caption{Where QCIVET sits relative to the four closest literature 
threads. The Classical, PQ-resistant, Q-augmented, and Q software 
contract columns refer to the literatures discussed in 
Sections~\ref{sec:related:classical}--\ref{sec:related:qse} respectively.}
\label{tab:positioning}
\begin{tabular}{|p{0.18\linewidth}|p{0.11\linewidth}|p{0.13\linewidth}|p{0.12\linewidth}|p{0.14\linewidth}|p{0.12\linewidth}|}
\hline
                          & \textbf{Classical}      & \textbf{PQ-resistant} & \textbf{Q-augmented} & \textbf{Q software contracts} & \textbf{QCIVET}        \\
\textbf{Question}         & in-toto / SLSA \cite{intoto, slsa}         & MBOM-PQC, ML-DSA \cite{ai_supply_chain_pqc}     & QHF / BB84 \cite{wang_qhf_blockchain, mozo_bb84_aes}          & Yamaguchi~\cite{yamaguchi_dbc_quantum}, Feng-Zhou~\cite{feng_zhou_subtyping} & (this work)           \\
\hline
Hybrid pipeline           & no                      & no                    & no                   & partial \newline (single circuit/module) & \textbf{yes}          \\
\hline
Hash-chain audit          & yes                     & yes                   & no                   & no                   & \textbf{yes}          \\
\hline
Semantic check            & no                      & no                    & no                   & yes \newline (state)          & \textbf{yes} \newline (observable) \\
\hline
Behavioural subtyping     & no                      & no                    & no                   & refinement \newline only~\cite{feng_zhou_refinement_orders}      & \textbf{yes} \newline (Liskov-Wing)~\cite{liskov_wing} \\
\hline
Real-time abort           & partial                 & no                    & no                   & yes \newline (intra-circuit)  & \textbf{yes} \newline (cross-stage) \\
\hline
Cloud-result audit        & no                      & no                    & no                   & no                   & \textbf{yes}          \\
\hline
\end{tabular}
\end{table}

In Table~\ref{tab:positioning}, ``semantic check'' means
verifying that quantum-stage outputs satisfy a calibrated
observable contract; ``hybrid pipeline'' means stage-level
coverage of a workflow with both classical and quantum stages;
and ``behavioural subtyping''~\cite{liskov_wing} means an
explicit Liskov-Wing substitutability discipline (rather than
refinement~\cite{feng_zhou_refinement_orders} or compliance).

\section{Background}
\label{sec:bg}

\subsection{Contracts and Behavioural Subtyping}

A contract for a class $A$ is a triple of invariants,
preconditions, and postconditions, denoted
$\spec(A)$~\cite{meyer_oosc, liskov_wing}. A class $B$ is a
behavioural subtype of $A$ ($B \preceq A$) if every operation
of $A$ that $B$ overrides preserves the externally visible
contract of $A$: $B$ does not strengthen the precondition and
does not weaken the postcondition. This is Liskov and Wing's
classical formulation~\cite{liskov_wing} and underlies modern
object-oriented programming (OOP) type discipline.

\subsection{Quantum States, Channels, and Observables}

We follow the conventions of~\cite{nielsen_chuang, preskill}. A pure
state of a single qubit is
$\ket{\psi} = \alpha\ket{0} + \beta\ket{1}$ with
$|\alpha|^2 + |\beta|^2 = 1$ and density matrix
$\rho = \ket{\psi}\bra{\psi}$. A general quantum operation is
a completely positive trace-preserving (CPTP) map
$\cE: \rho \mapsto \cE(\rho)$. The expectation of an
observable $O$ in the state $\rho$ is
$\langle O \rangle_\rho = \mathrm{Tr}(O \rho)$.

The diamond norm of a CPTP map $\cE$ is
$\diamondnorm{\cE} = \sup_\rho \trnorm{\cE(\rho)}$, extended to
differences by $\diamondnorm{\cE_A - \cE_B}$, which gives the
maximal distinguishability of two channels in the trace
norm~\cite{watrous_qit}. This is the natural distance for
behavioural-subtyping arguments~\cite{feng_zhou_subtyping} and
plays a central role in Section~\ref{sec:formal}.

Throughout the paper we use three norms on operators, all
specialisations of the Schatten $p$-norms (notation following
Watrous~\cite[Sec.~1.1]{watrous_qit}):
\begin{itemize}[leftmargin=*, itemsep=2pt]
\item \emph{Operator norm} $\opnorm{O}$ (Schatten $\infty$-norm):
  the largest singular value of $O$. It is the correct norm for
  an observable, since for $O$ self-adjoint it equals the
  largest absolute eigenvalue, i.e.\ the maximal observable
  outcome.
\item \emph{Trace norm} $\trnorm{X}$ (Schatten $1$-norm): the
  sum of singular values of $X$. For Hermitian $X$ this equals
  the sum of absolute eigenvalues. The trace norm is the
  correct distance between density operators and the natural
  norm for a state-level channel discrepancy
  $\cE_A(\rho) - \cE_B(\rho)$.
\item \emph{Diamond norm} $\diamondnorm{\cE}$, defined above:
  the operationally meaningful distance between channels,
  equal to the maximal trace distance between channel outputs
  taken over all (possibly entangled) inputs to
  $\cE \otimes \mathrm{id}$.
\end{itemize}
The three norms are linked by H{\"o}lder's inequality on
Hermitian operators,
$|\mathrm{Tr}(O\, X)| \le \opnorm{O}\, \trnorm{X}$, and by the
finite-dimensional equivalence between trace-norm and
diamond-norm bounds~\cite{paulsen_completely} (with a conversion factor at most $d$,
where $d := \dim \cH$). Section~\ref{sec:formal} formalises
these constants and uses them to prove soundness, completeness,
and compositionality of the contract framework.

For a bipartite system $\rho_{AB}$, the reduced state of $B$
is $\rho_B = \mathrm{Tr}_A(\rho_{AB})$. The reduced state
contains all marginals on $B$ but no phase information about
correlations with $A$.

\subsection{Hash Chains and External Anchors}

Following~\cite{schneier_kelsey, merkle}, an audit log is
\emph{tamper-evident} when each entry is computed as
$h_i = H(h_{i-1}\,\|\,\spec_i)$ for a collision-resistant hash
$H$, with $h_0$ a public genesis. A modification to $\spec_i$
that goes undetected requires producing a collision, which is
infeasible for SHA-256 under standard assumptions. To detect a
globally consistent rewrite (an adversary who replaces every
record after the fact), one binds the chain to an external
anchor: a publicly verifiable log such as Sigstore
Rekor~\cite{sigstore}, an RFC 3161 timestamp
authority~\cite{rfc3161}, or a public blockchain commitment.

\section{The QCIVET Framework}
\label{sec:framework}

\paragraph{Stages and spec records.}

A hybrid quantum--classical pipeline is a sequence of stages
$S_1, S_2, \dots, S_n$. Each stage $S_i$ has:
\begin{itemize}[leftmargin=*, itemsep=2pt, topsep=2pt]
\item a name (a string identifier);
\item a spec $\sigma_i$, a JSON-serialisable record of all
parameters that determine the stage's behaviour (transpiler
version, backend identifier, calibration snapshot hash,
classifier threshold, ansatz family, and so on);
\item optionally a list of observables measured during the
stage, each with a reference value and a calibrated tolerance
(only quantum stages carry these).
\end{itemize}
The spec is the unit of accountability: any change to the
behaviour of $S_i$ must reflect in $\sigma_i$.

\subsection{Subtyping for Stages}

Two stages $A$ and $B$ that purport to do ``the same thing''
(an override) live in a behavioural-subtyping relation. We
model each as a CPTP channel $\cE_A, \cE_B$ and introduce a
finite calibrated observable family $(\cO_A, \eps)$, where
$\cO_A = \{O_1, \dots, O_k\}$ is a set of self-adjoint
operators and $\eps \ge 0$ a tolerance.

\begin{definition}[Contract-preserving subtyping]
\label{def:cps}
We say $B \preceq_{(\cO_A, \eps)} A$ when, for every input
state $\rho$,
\[
\max_{O \in \cO_A}\,
\big|\,\mathrm{Tr}\!\bigl(O\,\cE_B(\rho)\bigr)
- \mathrm{Tr}\!\bigl(O\,\cE_A(\rho)\bigr)\,\big|
\le \eps.
\]
\end{definition}

When $\cO_A$ spans the operator basis of the underlying
Hilbert space (for instance all single-qubit Paulis
$\{X, Y, Z\}$ for a qubit), the test is informationally
complete; otherwise it is a projection onto the subspace of
contract-relevant observables.

\paragraph{Hash-chained audit trail.}

QCIVET maintains an audit log of triples
$(\sigma_i, h_{i-1}, h_i)$ with
\[
h_i = H\!\big(\,h_{i-1} \,\|\, \mathrm{canonical}(\sigma_i)\,\big),
\qquad h_0 = \texttt{0}^{64},
\]
where $\mathrm{canonical}$ is a deterministic JSON serialiser
(sorted keys, no whitespace) and $H$ is SHA-256. The chain is
appended to as the pipeline runs, and an external anchor
receives each new $h_i$ for tamper-evident timestamping.
Figure~\ref{fig:qcivet_workflow} illustrates how stage commits
stream to the QCIVET engine, which performs the hash-chain
check at every stage and the observable-deviation check at the
quantum stages, releasing the result on success and emitting
the audit trail on failure.

\begin{figure}[h]
\centering
\begin{tikzpicture}[
  font=\sffamily\footnotesize,
  stage/.style={
    draw, rounded corners=2pt, minimum width=1.7cm, minimum height=0.7cm,
    align=center, fill=blue!10, draw=blue!50, thick
  },
  qstage/.style={
    draw, rounded corners=2pt, minimum width=1.7cm, minimum height=0.7cm,
    align=center, fill=orange!15, draw=orange!70, thick, line width=0.8pt
  },
  engine/.style={
    draw, rounded corners=4pt, minimum width=11cm, minimum height=1.2cm,
    align=center, fill=green!8, draw=green!50, thick
  },
  arrow/.style={->, >=stealth, thick, draw=gray!70},
  hashlink/.style={->, >=stealth, dashed, draw=red!60, thick}
]

\node[font=\sffamily\small\bfseries] at (5.5, 1.0) 
  {Hybrid Quantum-Classical Pipeline};

\node[stage] (s1) at (0.0, 0)   {circuit\_def};
\node[stage] (s2) at (2.0, 0)   {transpile};
\node[stage] (s3) at (4.0, 0)   {backend\_sel};
\node[qstage] (s4) at (6.0, 0)  {calibration};
\node[qstage] (s5) at (8.0, 0)  {execution};
\node[stage] (s6) at (10.0, 0)  {meas\_output};

\draw[arrow] (s1) -- (s2);
\draw[arrow] (s2) -- (s3);
\draw[arrow] (s3) -- (s4);
\draw[arrow] (s4) -- (s5);
\draw[arrow] (s5) -- (s6);

\draw[hashlink] (s1.south) -- (s1.south |- 0,-1.4);
\draw[hashlink] (s2.south) -- (s2.south |- 0,-1.4);
\draw[hashlink] (s3.south) -- (s3.south |- 0,-1.4);
\draw[hashlink] (s4.south) -- (s4.south |- 0,-1.4);
\draw[hashlink] (s5.south) -- (s5.south |- 0,-1.4);
\draw[hashlink] (s6.south) -- (s6.south |- 0,-1.4);

\node[engine] (engine) at (5.0, -2.1) 
  {\textbf{QCIVET Engine}\\[2pt]
   \scriptsize Hash chain (every stage) $\;\bullet\;$ 
   Observable-deviation check (quantum stages) $\;\bullet\;$ 
   External anchor};

\node[draw, rounded corners=2pt, fill=gray!10, draw=gray!50,
      font=\sffamily\scriptsize, align=center,
      minimum width=10cm, minimum height=0.6cm] (output) at (5.0, -3.5)
  {OK $\to$ release result \quad $\bullet$ \quad 
   FAIL $\to$ abort and emit audit trail};

\draw[arrow] (engine.south) -- (output.north);

\node[font=\sffamily\scriptsize, align=center] at (5.0, -4.3)
  {%
   \tikz\draw[fill=blue!10, draw=blue!50] (0,0) rectangle (0.3,0.2);\,classical stage \quad
   \tikz\draw[fill=orange!15, draw=orange!70, line width=0.8pt] (0,0) rectangle (0.3,0.2);\,quantum stage \quad
   \tikz\draw[->, >=stealth, dashed, draw=red!60, thick] (0,0.1) -- (0.5,0.1);\,hash commit
  };

\end{tikzpicture}
\caption{QCIVET workflow. A hybrid quantum-classical pipeline
streams stage commits to the QCIVET engine, which performs
hash-chained syntactic integrity checks at every stage and
calibrated observable-deviation checks at the quantum stages
(orange). On success, the result is released to the customer;
on failure, the engine aborts and emits the audit trail.}
\label{fig:qcivet_workflow}
\end{figure}
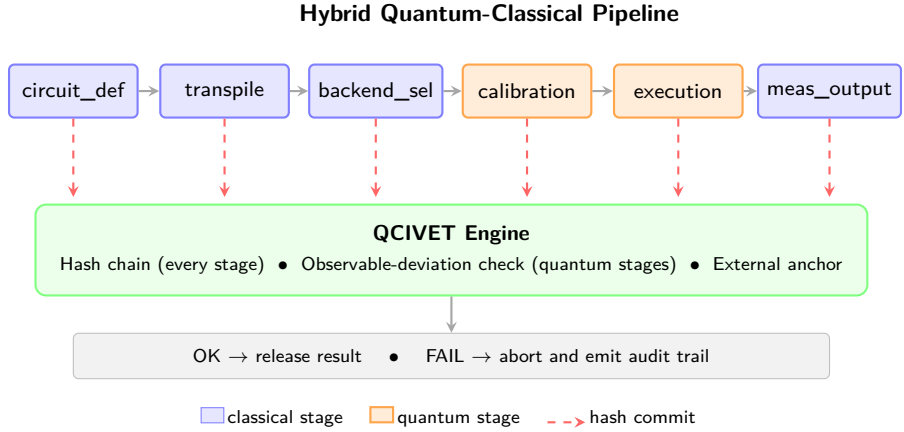

\FloatBarrier

\subsection{Real-Time Verification Engine}

The \texttt{IntegrityVerifier} class
(Listing~\ref{lst:engine}) exposes a single
\texttt{commit\_stage} method that the host pipeline calls when
each stage completes. A commit performs four operations:
\begin{enumerate}[leftmargin=*, itemsep=2pt, topsep=2pt]
\item check the observables (if any) against
$(\cO_A, \eps)$, raising \texttt{IntegrityViolation}
(kind=\texttt{observable}) on violation;
\item compute $h_i$ from $h_{i-1}$ and $\sigma_i$;
\item verify chain-head consistency, raising
\texttt{IntegrityViolation} (kind=\texttt{hash}) on a desync;
\item submit $h_i$ to the external anchor, raising
\texttt{IntegrityViolation} (kind=\texttt{anchor}) on
submission failure.
\end{enumerate}
On any violation the host pipeline catches the exception and
aborts before launching downstream stages. The cost of a
commit, measured in Section~\ref{sec:rt}, is below 0.1 ms on
commodity hardware. Globally consistent rewrites that escape
the local chain check are caught by a separate
\texttt{verify\_against\_anchor} method that confirms the
local chain appears as a contiguous block in the anchor log.

\begin{lstlisting}[caption={Pipeline-host pattern using QCIVET.},
                   label={lst:engine}]
verifier = IntegrityVerifier(anchor=ExternalAnchor("rekor.log"))
try:
    for stage in pipeline:
        out = stage.run()
        verifier.commit_stage(StageResult(
            name=stage.name,
            spec=stage.spec_dict(),
            observables=stage.observables(),  # only on quantum stages
        ))
except IntegrityViolation as e:
    log.error(f"halted at stage {e.stage_index} ({e.kind}): {e}")
    pipeline.abort()
\end{lstlisting}

\paragraph{Sneaky overrides.}

A subtle threat is the override that passes a weak contract
but fails an informationally complete one. Take
$A = R_y(\theta)$ and $B = R_y(\theta) \cdot S$, where $S$ is
the phase gate. On the input $\ket{0}$, both prepare states
with the same $\langle Z \rangle$, so $B$ is indistinguishable
from $A$ on $\cO_A = \{Z\}$. But the two channels differ on
$X$ and $Y$ by an amount up to $|\sin\theta|$, and an
$\cO_A = \{X, Y, Z\}$ contract immediately exposes $B$. A
sufficient condition for the existence of such overrides is
given by Proposition~\ref{prop:sneaky}.

\section{Formal Properties}
\label{sec:formal}

This section makes Definition~\ref{def:cps} mathematically
precise. We give a soundness theorem (channel-level closeness
implies observable-level closeness), a conditional
completeness result (the converse holds when the observable
family is informationally complete), and a compositionality
theorem (contract preservation transports along inheritance
chains). We then characterise the sneaky-subtype failure mode
introduced informally in Section~4.2 and clarify the
relationship between our setting and the foundational
refinement orders of Feng et al.~\cite{feng_zhou_refinement_orders,feng_zhou_subtyping},
showing how the two frameworks complement each other.

\subsection{Setting and Notation}

Throughout this section we work with finite-dimensional
Hilbert spaces. Let $\cH$ have dimension
$d := \dim \cH < \infty$. For an $n$-qubit system,
$d = 2^n$; the single-qubit case is $d = 2$. Let
$\cD(\cH)$ denote the set of density operators on $\cH$. For
two CPTP channels $\cE_A, \cE_B : \cD(\cH) \to \cD(\cH)$, the
diamond norm of their difference is
\begin{equation}
\diamondnorm{\cE_A - \cE_B} = \sup_{\rho \in \cD(\cH \otimes \cH)}
\trnorm{(\cE_A \otimes \mathrm{id})(\rho) - (\cE_B \otimes \mathrm{id})(\rho)}.
\end{equation}
The diamond norm is the operationally meaningful distance
between channels: it equals the maximal trace distance
achievable by any input, possibly entangled with an
ancilla~\cite{watrous_qit}.

For an observable $O$ (self-adjoint on $\cH$), $\opnorm{O}$
denotes the operator norm and $\trnorm{X}$ denotes the trace
norm, both as introduced in Section~\ref{sec:bg}. We write
$K_A := \max_{O \in \cO_A} \opnorm{O}$ for the tight spectrum
bound of the contract observable family $\cO_A$.

\begin{definition}[Informationally complete observable family]
\label{def:ic}
A family $\cO \subset \mathrm{Herm}(\cH)$ is informationally
complete if its real linear span equals $\mathrm{Herm}(\cH)$
itself; equivalently, for any two states
$\rho_1, \rho_2 \in \cD(\cH)$,
$\mathrm{Tr}(O \rho_1) = \mathrm{Tr}(O \rho_2)$ for all
$O \in \cO$ implies $\rho_1 = \rho_2$.
\end{definition}

For a single qubit ($d = 2$), the Pauli set $\{X, Y, Z\}$ is
informationally complete; the singleton $\{Z\}$ is not.

\subsection{Soundness}

\begin{theorem}[Soundness]
\label{thm:soundness}
Let $A, B$ be classes with channels $\cE_A, \cE_B$ and let
$\cO_A$ be a contract observable family with tight spectrum
bound $K_A$. If $\diamondnorm{\cE_A - \cE_B} \le \delta$, then
for every input set $\cS \subseteq \cD(\cH)$,
\begin{equation}
B \preceq_{(\cO_A,\,K_A \delta),\, \cS} A.
\end{equation}
\end{theorem}

\begin{proof}
For any $\rho \in \cS$ and $O \in \cO_A$, by H{\"o}lder's
inequality on Hermitian operators,
\begin{align}
\bigl| \mathrm{Tr}(O\, \cE_B(\rho)) - \mathrm{Tr}(O\, \cE_A(\rho)) \bigr|
&= \bigl| \mathrm{Tr}(O\, [\cE_B(\rho) - \cE_A(\rho)]) \bigr| \\
&\le \opnorm{O}\, \trnorm{\cE_B(\rho) - \cE_A(\rho)} \\
&\le K_A\, \diamondnorm{\cE_A - \cE_B} \le K_A \delta.
\end{align}
Taking suprema over $\rho \in \cS$ and $O \in \cO_A$ gives the
claim.
\end{proof}

\subsection{Conditional Completeness}

\begin{theorem}[Conditional completeness]
\label{thm:completeness}
Suppose $\cO_A$ is informationally complete in
$\mathrm{Herm}(\cH)$, and let $\cS \subseteq \cD(\cH)$ be a
state set whose affine span is the full set of density
operators. Then there exists a constant $C(\cO_A) > 0$,
depending only on the observable family, such that for any
CPTP channels $\cE_A, \cE_B$,
\begin{equation}
B \preceq_{(\cO_A,\,\eps),\, \cS} A
\Longrightarrow \diamondnorm{\cE_A - \cE_B} \le C(\cO_A)\, \eps.
\end{equation}
For a single qubit with $\cO_A = \{X, Y, Z\}$ and
$\cS \supseteq \{\ket{0},\ket{1},\ket{+},\ket{-}\}$, one may
take $C(\cO_A) = 2\sqrt{2}$.
\end{theorem}

\begin{proof}
Define the linear map
$\Delta : \mathrm{Herm}(\cH) \to \mathrm{Herm}(\cH)$ by
$\Delta(\rho) = \cE_A(\rho) - \cE_B(\rho)$. The hypothesis
gives $\sup_{O \in \cO_A} |\mathrm{Tr}(O\, \Delta(\rho))| \le
\eps$ for every $\rho \in \cS$. Since $\cO_A$ spans
$\mathrm{Herm}(\cH)$, the family of linear functionals
$\sigma \mapsto \mathrm{Tr}(\sigma O)$ is norming on
$\mathrm{Herm}(\cH)$; equivalently, there is a constant
$c(\cO_A) > 0$ with
\[
\trnorm{\sigma} \le c(\cO_A)\, \sup_{O \in \cO_A}
|\mathrm{Tr}(O\, \sigma)|
\quad \forall \sigma \in \mathrm{Herm}(\cH),
\]
because all norms on a finite-dimensional space are equivalent.
Apply this to $\sigma = \Delta(\rho)$ for $\rho \in \cS$:
$\trnorm{\Delta(\rho)} \le c(\cO_A) \eps$. Linearity of
$\Delta$ extends the bound to every $\rho \in \cD(\cH)$, with
an inflated constant. Stinespring purification of
$\cE_A - \cE_B$ then converts the state-level trace-norm bound
into a diamond-norm bound on the channel difference, again
with a finite constant depending only on $d$ and on
$c(\cO_A)$~\cite{watrous_qit}.

For the single-qubit Pauli-eigenstate case the four states
span $\mathrm{Aff}(\cD(\cH))$ and the constant is computable
directly: any traceless Hermitian $\sigma$ on $\mathbb{C}^2$
admits
$\sigma = \tfrac{1}{2}\sum_{P \in \{X,Y,Z\}} \mathrm{Tr}(P\sigma) P$,
and $\trnorm{P} = 2$ for each Pauli, so
$\trnorm{\sigma} \le \sqrt{2}\, \sup_P |\mathrm{Tr}(P\sigma)|$.
The diamond-to-trace conversion factor is at most $2$ in
dimension~$2$. Multiplying gives $C(\cO_A) \le 2\sqrt{2}$.
\end{proof}

\begin{remark}[Scaling with dimension]
\label{rem:dim_scaling}
The constant $C(\cO_A)$ depends both on the geometry of the
observable family (the factor $c(\cO_A)$ from the
norming-functional argument) and on the dimension $d$ of the
Hilbert space (from the trace-norm to diamond-norm conversion).
For an $n$-qubit system, $d = 2^n$, and the diamond conversion
factor is at most $d$ in the worst case. If $\cO_A$ is taken
to be the full $n$-qubit Pauli family
($4^n - 1$ non-identity tensor products of $\{I, X, Y, Z\}$),
the same Pauli decomposition argument yields
$c(\cO_A) \le \sqrt{4^n - 1}$, so $C(\cO_A)$ grows
polynomially in $d$ but exponentially in $n$. In practice,
deployers do not use the full Pauli family: they choose a
small structured subset
(e.g., the Pauli set on a single relevant qubit, or the
operator that defines the application-specific reference value).
The calibration data of Section~\ref{sec:exp:disc} reports the
\emph{empirically} observed constant for the contracts used in
this paper; the worst-case theoretical scaling above is an
upper bound, not a typical value.
\end{remark}

\subsection{Compositionality}

\begin{theorem}[Compositionality]
\label{thm:composition}
Let $A_1, A_2$ be supertypes with respective contract families
$\cO_{A_1}, \cO_{A_2}$ and spectrum bounds $K_1, K_2$. Suppose
$B_1 \preceq_{(\cO_{A_1},\,\eps_1)} A_1$ and
$B_2 \preceq_{(\cO_{A_2},\,\eps_2)} A_2$. Define
$A := \cE_{A_2} \circ \cE_{A_1}$ and
$B := \cE_{B_2} \circ \cE_{B_1}$, with outer contract
$\cO_A := \cO_{A_2}$. Then for every $\rho \in \cD(\cH)$ and
$O \in \cO_A$,
\begin{equation}
\bigl| \mathrm{Tr}(O\, \cE_B(\rho)) - \mathrm{Tr}(O\, \cE_A(\rho)) \bigr|
\le \eps_2 + K_2\, c(\cO_{A_1})\, \eps_1.
\end{equation}
\end{theorem}

\begin{proof}
We use the standard ``add-and-subtract'' decomposition: insert
the auxiliary term $\cE_{A_2}(\cE_{B_1}(\rho))$ to obtain
\begin{align}
\cE_B(\rho) - \cE_A(\rho)
  &= \cE_{B_2}(\cE_{B_1}(\rho)) - \cE_{A_2}(\cE_{A_1}(\rho)) \\
  &= \underbrace{\bigl[\cE_{B_2}(\cE_{B_1}(\rho))
                       - \cE_{A_2}(\cE_{B_1}(\rho))\bigr]}_{\text{(I): inner-stage discrepancy at the second stage}}
   + \underbrace{\cE_{A_2}\bigl[\cE_{B_1}(\rho)
                       - \cE_{A_1}(\rho)\bigr]}_{\text{(II): first-stage error propagated through } \cE_{A_2}}.
\end{align}
Fix any $O \in \cO_{A_2}$ with $\opnorm{O} \le K_2$ and bound
each term separately.

\paragraph{Term (I).}
Both branches share the input
$\cE_{B_1}(\rho) \in \cD(\cH)$, so the second-stage hypothesis
$B_2 \preceq_{(\cO_{A_2},\,\eps_2)} A_2$ applies directly to
this input:
\[
\bigl| \mathrm{Tr}\bigl(O\, [\cE_{B_2}(\cE_{B_1}(\rho))
                            - \cE_{A_2}(\cE_{B_1}(\rho))]\bigr) \bigr|
\le \eps_2.
\]

\paragraph{Term (II).}
By H{\"o}lder's inequality on Hermitian operators,
\[
\bigl| \mathrm{Tr}\bigl(O\, \cE_{A_2}[\cE_{B_1}(\rho)
                                       - \cE_{A_1}(\rho)]\bigr) \bigr|
\le \opnorm{O}\, \trnorm{\cE_{A_2}[\cE_{B_1}(\rho)
                                     - \cE_{A_1}(\rho)]}.
\]
Since $\cE_{A_2}$ is a CPTP map, it is a trace-norm contraction
on Hermitian operators~\cite[Sec.~3.3]{watrous_qit}, so
$\trnorm{\cE_{A_2}(Y)} \le \trnorm{Y}$ for every Hermitian $Y$.
Applying this with
$Y = \cE_{B_1}(\rho) - \cE_{A_1}(\rho)$,
\[
\trnorm{\cE_{A_2}[\cE_{B_1}(\rho) - \cE_{A_1}(\rho)]}
\le \trnorm{\cE_{B_1}(\rho) - \cE_{A_1}(\rho)}.
\]
The first-stage hypothesis
$B_1 \preceq_{(\cO_{A_1},\,\eps_1)} A_1$ together with the
norming-functional argument inside the proof of
Theorem~\ref{thm:completeness} (the state-level trace-norm
bound, applied \emph{before} the diamond-to-trace conversion)
yields
$\trnorm{\cE_{B_1}(\rho) - \cE_{A_1}(\rho)}
\le c(\cO_{A_1})\, \eps_1$ for every
$\rho \in \cD(\cH)$.\footnote{Equivalently, one may use the
diamond-norm bound $\diamondnorm{\cE_{A_1} - \cE_{B_1}}
\le C(\cO_{A_1})\, \eps_1$ from Theorem~\ref{thm:completeness}
and then specialise to a single input $\rho$, absorbing the
diamond-to-trace conversion into the constant; the
resulting bound has the same $O(\eps_1)$ scaling.}
Combining,
\[
\bigl| \mathrm{Tr}\bigl(O\, \text{(II)}\bigr) \bigr|
\le K_2\, c(\cO_{A_1})\, \eps_1.
\]

\paragraph{Combining.}
By the triangle inequality on the two terms,
\[
\bigl| \mathrm{Tr}(O\, \cE_B(\rho))
       - \mathrm{Tr}(O\, \cE_A(\rho)) \bigr|
\le \eps_2 + K_2\, c(\cO_{A_1})\, \eps_1,
\]
which is the claimed bound.
\end{proof}

\begin{remark}[Practical reading]
Theorem~\ref{thm:composition} says errors compose roughly
additively along an inheritance chain. A pipeline of $n$ valid
overrides each at tolerance $\eps$ has worst-case observable
deviation $O(n\eps)$ at the end. This is the kind of
book-keeping a static analyser or audit tool must track in
practice.
\end{remark}

\subsection{The Sneaky Subtype, Formally}

\begin{proposition}[Sneakiness characterisation]
\label{prop:sneaky}
Let $\cE_A, \cE_B$ be unitary single-qubit channels and let
$\cO_A \subsetneq \mathrm{Herm}(\mathbb{C}^2)$. There exists a
sneaky subtype $\cE_B \neq \cE_A$ with
$B \preceq_{(\cO_A,\,0)} A$ if and only if $\cO_A$ is not
informationally complete.
\end{proposition}

\begin{proof}
($\Leftarrow$) Suppose $\cO_A$ is not informationally
complete. By Definition~\ref{def:ic}, there exist distinct
states $\rho_1 \neq \rho_2$ with
$\mathrm{Tr}(O\rho_1) = \mathrm{Tr}(O\rho_2)$ for all
$O \in \cO_A$. The phase gate $S$ provides an explicit
witness: $S$ commutes with $Z$, so
$\mathrm{Tr}(Z \rho) = \mathrm{Tr}(Z S\rho S^\dagger)$.
Setting $\cE_B := \cE_A \circ \mathrm{Ad}_S$ gives
$B \preceq_{(\{Z\},\,0)} A$ with $\cE_B \neq \cE_A$.

($\Rightarrow$) If $\cO_A$ is informationally complete, then
by Theorem~\ref{thm:completeness},
$B \preceq_{(\cO_A,\,0)} A$ forces
$\diamondnorm{\cE_A - \cE_B} = 0$ and hence $\cE_A = \cE_B$.
\end{proof}

\subsection{Comparison with Refinement Orders}

Our framework stands in a direct and constructive
relationship with the refinement calculus for quantum programs
of Feng et al.~\cite{feng_zhou_refinement_orders,feng_zhou_subtyping}.
They define $\cE_B \sqsubseteq \cE_A$ on CPTP and completely
positive trace non-increasing (CPTN) maps via the
complete-positivity ordering of the corresponding
super-operators, equivalent to a diamond-norm comparison up to
constants in finite dimension. Our observable-restricted
condition (Definition~\ref{def:cps}) is the operational
projection of this ordering onto a finite, calibration-ready
observable interface: Theorem~\ref{thm:soundness} shows that
the Feng et al.\ ordering implies ours;
Theorem~\ref{thm:completeness} shows that the converse holds
when the observable family is informationally complete; and
Proposition~\ref{prop:sneaky} characterises precisely when the
projection is lossy, which is exactly the sneaky-subtype
regime. Read this way, the two frameworks form a tight pair:
the refinement orders of Feng et al.\ give the full
denotational picture of substitutability, and our framework
gives the runtime-observable, hardware-evaluable projection of
that picture, which is the right object for an OOP-style
contract because it respects encapsulation: a class commits
only to its declared observables, not to its full channel
realisation.

\section{Experimental Validation}
\label{sec:exp}

This section instantiates the framework on a single-qubit
example and validates it under three increasingly realistic
settings: synthetic depolarising noise, the calibrated noise
models of two production IBM Quantum processors (FakeBrisbane
and FakeFez), and an end-to-end run on the real
\texttt{ibm\_fez} (Heron r2) processor accessed through the
IBM Quantum cloud.

\subsection{Reference Class and Three Subtypes}
\label{sec:exp:setup}

We fix the single-qubit supertype
\begin{equation}
A: \cE_A(\rho) = R_y(\theta)\, \rho\, R_y(\theta)^\dagger,
\qquad \theta = 2\pi/5,
\end{equation}
and study three candidate subtypes:
\begin{itemize}[leftmargin=*, itemsep=2pt]
\item $B_{\mathrm{good}}$: $\cE_{B_g}$ is the identity composition
$S R_x(\theta) S^\dagger$ acting by conjugation, which equals
$\cE_A$ via $R_y(\theta) = S R_x(\theta) S^\dagger$;
\item $B_{\mathrm{bad}}(\delta)$:
$\cE_{B_b}(\rho) = R_y(\theta+\delta)\, \rho\, R_y(\theta+\delta)^\dagger$,
an over-rotation by $\delta = 0.4$;
\item $B_{\mathrm{sneaky}}$:
$\cE_{B_s}(\rho) = S\, \cE_A(\rho)\, S^\dagger$, which
preserves $\langle Z \rangle$ exactly but flips the sign of
$\langle Y \rangle$.
\end{itemize}
The contract observable families considered are the
informationally complete $\cO_A^{\text{full}} = \{X, Y, Z\}$
and the weak $\cO_A^{\text{weak}} = \{Z\}$.

\paragraph{Inputs.}
\label{sec:exp:inputs}

We test on six input states: the four Pauli eigenstates
$\{\ket{0}, \ket{1}, \ket{+}, \ket{-}\}$ plus two off-axis
states $\ket{\psi_1} = R_z(1.3) R_y(0.7)\ket{0}$ and
$\ket{\psi_2} = R_z(0.4) R_y(2.1)\ket{0}$. The off-axis pair
together with the generic angle $\theta = 2\pi/5$ ensures that
all three Pauli expectations of the supertype are simultaneously
non-trivial on every input. This avoids ``accidental zeros''
that would be ambiguous under hardware noise (a noisy estimator
returning $0.00$ could either track a genuine zero or have lost
its signal entirely), and lets us use a single, coherent
instantiation across both the noiseless analyses
(Experiments~1--4) and the device-noise validation
(Experiments~5--6).

\subsection{Experiment 1: Subtype Separation in the Ideal Setting}
\label{sec:exp1}

\paragraph{Methodology.} The script
\texttt{quantum\_oop\_\allowbreak simulation.py} constructs the supertype
$A = R_y(2\pi/5)$ and the three candidate subtypes
$B_{\mathrm{good}}$, $B_{\mathrm{bad}}(0.4)$, $B_{\mathrm{sneaky}}$
exactly as defined in Section~\ref{sec:exp:setup}, evaluated on
the six inputs of Section~\ref{sec:exp:inputs}. For each
candidate, the script computes the analytic output state on
every input via Qiskit's \texttt{Statevector}, evaluates the
three Pauli expectations exactly via $\mathrm{Tr}(\rho O)$, and
reports the worst-case deviation from the supertype values. No
sampling is involved at this stage; the deviations are the exact
operator-norm differences.

\begin{table}[H]
\centering
\caption{Experiment 1, ideal noiseless setting. Worst-case
observable deviation across the six input states. Generated
by \texttt{quantum\_oop\_\allowbreak simulation.py}, function
\texttt{experiment\_1} (footnote~\ref{fn:repo}).}
\label{tab:exp1}
\begin{tabular}{l c c}
\toprule
Candidate & worst $\{X,Y,Z\}$ & worst $\{Z\}$ \\
\midrule
$B_{\mathrm{good}}$        & 0.000 & 0.000 \\
$B_{\mathrm{bad}}(0.4)$    & 0.395 & 0.395 \\
$B_{\mathrm{sneaky}}$      & 1.401 & 0.000 \\
\bottomrule
\end{tabular}
\end{table}

For each candidate $B$ we compute, for every input state and
every Pauli observable, the noiseless deviation
$\Delta_O^B(\rho) = |\mathrm{Tr}(O\, \cE_B(\rho))
- \mathrm{Tr}(O\, \cE_A(\rho))|$.
Table~\ref{tab:exp1} reports the worst-case deviation across
the six inputs for each of the two contract families.

The valid override $B_{\mathrm{good}}$ is at zero deviation
under both contracts. The over-rotation $B_{\mathrm{bad}}$
violates both. $B_{\mathrm{sneaky}}$ is the formal sneaky case
of Proposition~\ref{prop:sneaky}: it slips under the weak
$\{Z\}$-only contract while being maximally exposed by the
full one. Figure~\ref{fig:exp1} shows this per-(input, observable) 
pattern, with $B_{\mathrm{sneaky}}$ indistinguishable from $A$ 
on $Z$ alone but maximally exposed by $X$ and $Y$.

\begin{figure}[H]
\centering
\includegraphics[width=0.75\linewidth]{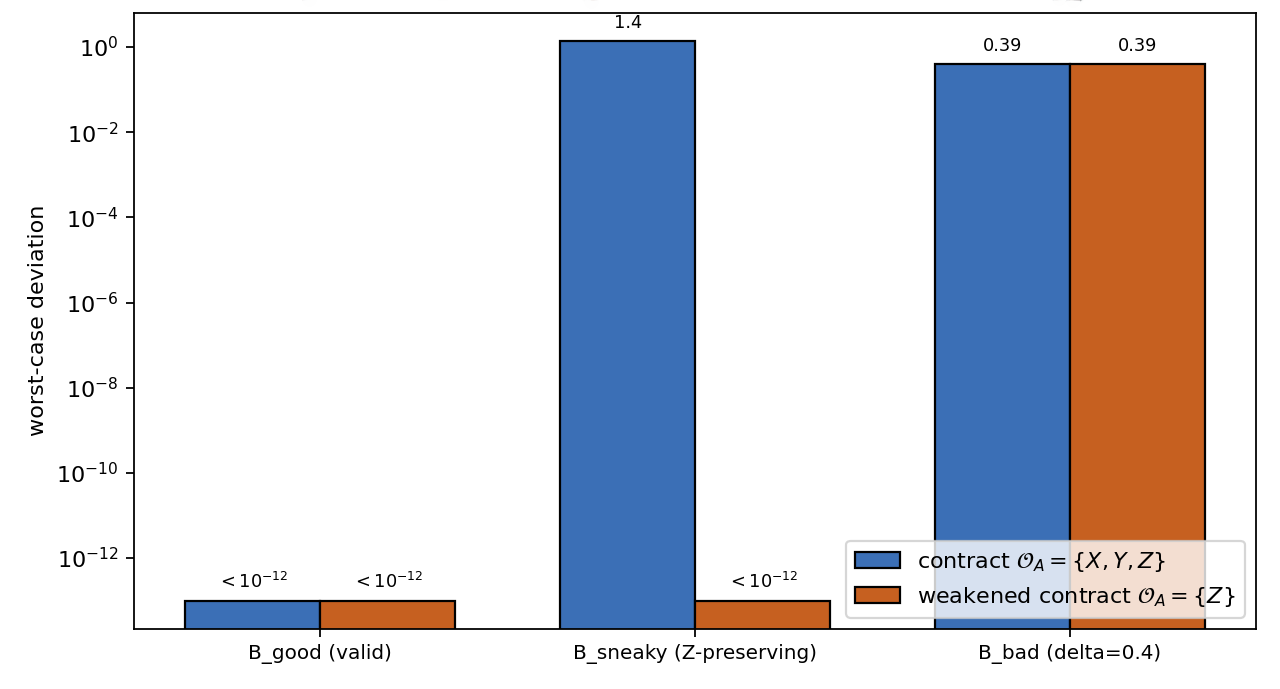}
\caption{Experiment~1. Per-(input, observable) deviation for
the three candidates in the noiseless setting. The pattern is
the operational fingerprint of
Proposition~\ref{prop:sneaky}: $B_{\mathrm{sneaky}}$ is
indistinguishable from $A$ on $Z$ alone, but maximally
exposed by $X$ and $Y$. Reproduced by
\texttt{quantum\_oop\_\allowbreak simulation.py}, function
\texttt{plot\_experiment\_1}
(footnote~\ref{fn:repo}).}
\label{fig:exp1}
\end{figure}

\FloatBarrier
\subsection{Experiment 2: Two-Qubit Reduced State}

\paragraph{Methodology.} The script
\texttt{quantum\_oop\_\allowbreak simulation.py} (function
\texttt{experiment\_2}) constructs the post-CNOT
(controlled-NOT) entangled
state $\ket{\Psi_{AB}} = \alpha\ket{00} + \beta\ket{11}$ for
ten random amplitude pairs, computes the partial trace over
qubit $A$ analytically via Qiskit's \texttt{partial\_trace}
primitive, and compares the result against the diagonal
matrix $\mathrm{diag}(|\alpha|^2, |\beta|^2)$ in Frobenius
norm. The deviation is reported per trial and as a maximum
across trials. This serves as a sanity check for the
behavioural-subtyping interpretation of encapsulation: the
reduced state carries only marginal probabilities, with all
correlation information traced away.

We numerically verify the partial-trace identity
$\rho_B = \mathrm{Tr}_A(\ket{\Psi_{AB}}\bra{\Psi_{AB}})
= |\alpha|^2 \ket{0}\bra{0} + |\beta|^2 \ket{1}\bra{1}$
on the post-CNOT state
$\ket{\Psi_{AB}} = \alpha\ket{00} + \beta\ket{11}$ for ten
random pairs $(\alpha, \beta)$. The Frobenius distance between
the analytical reduced state and the numerically computed
partial trace is bounded by $3.7 \times 10^{-33}$ in every
trial, confirming that the reduced state carries marginal
probabilities only.\footnote{Computed by
\texttt{quantum\_oop\_\allowbreak simulation.py}, function
\texttt{experiment\_2}; see footnote~\ref{fn:repo}.}

\FloatBarrier
\subsection{Experiment 3: Synthetic Depolarising Noise}

\paragraph{Methodology.} The script attaches a depolarising
channel of parameter $p$ to every single-qubit gate using
Qiskit Aer's
\texttt{depolarizing\_error}
and \texttt{NoiseModel} primitives.
The script sweeps $p$ over the eleven values
\begin{equation*}
p \in \{0, 0.001, 0.002, 0.005, 0.01, 0.02, 0.03, 0.05, 0.07, 0.10\},
\end{equation*}
and, for each $p$, runs $B_{\mathrm{good}}$
under the resulting noisy simulator with 4096 shots per
measurement, repeated for 20 independent trials. The
$X$, $Y$, $Z$ expectations are estimated from counts after
applying the appropriate basis-change rotations
($H$ for $X$; $S^\dag H$ for $Y$; identity for $Z$). The
script then aggregates each trial's worst-case deviation
across the three Paulis, and reports the mean, standard
deviation, and 95th percentile across trials.

We attach a single-parameter depolarising channel
$\cN_p(\rho) = (1-p)\rho + p\, I/2$ to each gate and sweep
$p$ over $[0, 0.05]$. For $B_{\mathrm{good}}$ the
worst-case observable deviation $\eps(p)$ averaged over inputs
is near-linear, $\eps(p) \approx 1.6\, p$ for small $p$; see
Figure~\ref{fig:exp3}. This calibrates a synthetic tolerance
budget but is not predictive of behaviour on a real device,
which motivates the device-derived experiments below.

\begin{figure}[H]
\centering
\includegraphics[width=0.7\linewidth]{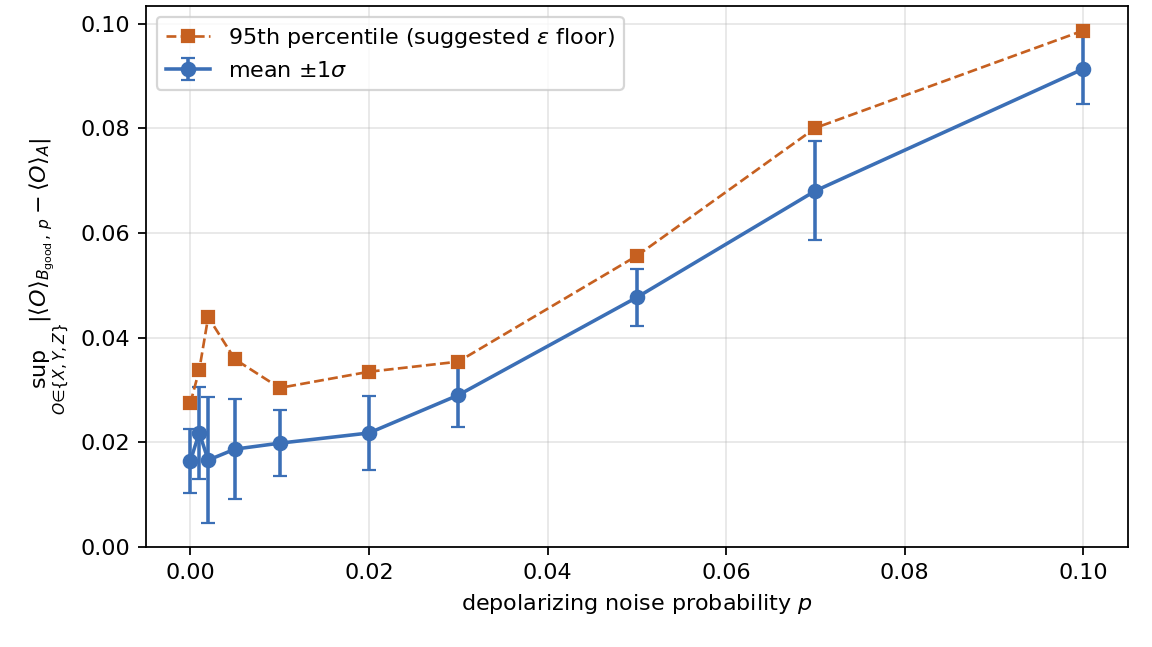}
\caption{Experiment~3. Synthetic depolarising-noise
calibration of $\eps$ for $B_{\mathrm{good}}$. The relation
is approximately linear in $p$ for small $p$. Reproduced by
\texttt{quantum\_oop\_\allowbreak simulation.py}, function
\texttt{plot\_experiment\_3}
(footnote~\ref{fn:repo}).}
\label{fig:exp3}
\end{figure}

\FloatBarrier
\subsection{Experiment 4: $\delta$-Sweep for $B_{\mathrm{bad}}$}
\label{sec:exp4}

\paragraph{Methodology.} The script sweeps $\delta$ over the
nine values $\{0, 0.01, 0.02, 0.05, 0.1, 0.2, 0.4, 0.6, 0.8\}$
and, for each $\delta$, constructs
$B_{\mathrm{bad}}(\delta) = R_y(2\pi/5+\delta)$
and computes the worst-case noiseless deviation from the
supertype across all six inputs and all three Pauli observables.
This is done analytically (no sampling) using
\texttt{Statevector} and trace formulas, so the resulting
deviations are the exact operator norms. The purpose is to
relate physical perturbation magnitude (radians) to the
observable-space deviation that a tolerance threshold has to
discriminate.

We sweep the over-rotation parameter $\delta$ from $0$ to
$\pi/2$ and report worst-case deviation under both contracts;
see Figure~\ref{fig:exp4}.
$\Delta^{B_{\mathrm{bad}}}_{\{X,Y,Z\}}(\delta)$ tracks
$\sin\delta$, while
$\Delta^{B_{\mathrm{bad}}}_{\{Z\}}(\delta)$ tracks
$|\cos\theta - \cos(\theta+\delta)|$. A tolerance choice
$\eps = 0.05$ corresponds to a detection threshold
$\delta^\star \approx 0.05$ radians.

\begin{figure}[t]
\centering
\includegraphics[width=0.7\linewidth]{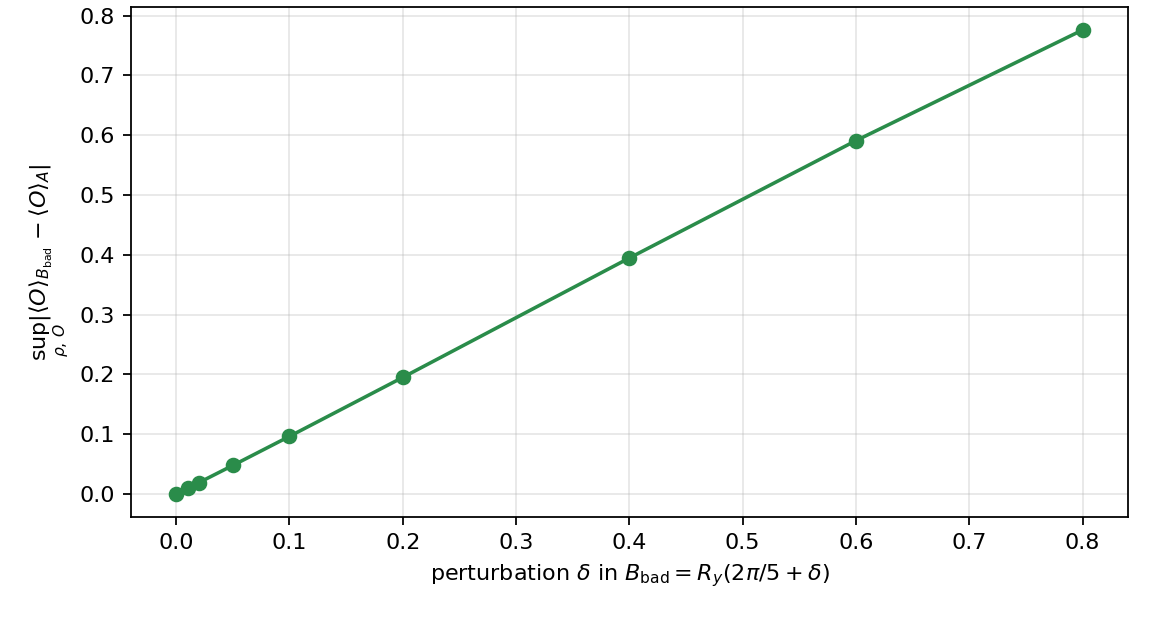}
\caption{Experiment~4. $\delta$-sweep for $B_{\mathrm{bad}}$.
The full $\{X,Y,Z\}$ contract is sensitive to small
$\delta$; the $\{Z\}$-only contract is much less so.
Reproduced by
\texttt{quantum\_oop\_\allowbreak simulation.py}, function
\texttt{plot\_experiment\_4}
(footnote~\ref{fn:repo}).}
\label{fig:exp4}
\end{figure}

\subsection{Experiment 5: $\eps$ Calibration under Realistic Device Noise}
\label{sec:exp5}

\paragraph{Methodology.} The script
\texttt{quantum\_oop\_\allowbreak device\_\allowbreak validation.py} replaces the
synthetic depolarising channel of Experiment~3 with the
calibration data published by IBM for the corresponding real
backends, accessed via
\texttt{NoiseModel.from\_backend()} on the
\texttt{FakeBrisbane} and \texttt{FakeFez} backends from
\texttt{qiskit\_ibm\_runtime.fake\_provider}. These fake
backends embed the actual T1, T2, gate error and readout error
data of the corresponding physical processors, so deviations
under them are representative of what would be observed on the
real hardware. For each (device, input state) pair the script
runs $B_{\mathrm{good}}$ in batched mode (one big sampler job
of $20 \times 4096$ shots, the per-shot outcomes split back
into 20 trials of 4096 shots each), measures
$\langle Z\rangle$ on every shot, and aggregates the
per-trial deviation against the noiseless reference value of
$\langle Z\rangle_A$.

We replace the synthetic depolarising channel with the
calibration-derived noise models of two production IBM
Quantum processors:
\begin{itemize}[leftmargin=*, itemsep=2pt]
\item \texttt{FakeBrisbane} (Eagle r3, 127 qubits, single-qubit
gate error $\sim 3 \times 10^{-4}$);
\item \texttt{FakeFez} (Heron r2, 156 qubits, single-qubit
gate error $\sim 2 \times 10^{-4}$).
\end{itemize}
For each device, $B_{\mathrm{good}}$ is run on each of the six
input states with twenty trials of 4096 shots, and the
per-input deviation $\eps$ is reported as mean and 95th
percentile across trials (Table~\ref{tab:exp5},
Figure~\ref{fig:exp5}).

The Heron r2 device produces deviations roughly half those of
Eagle r3, consistent with the published improvement in
single-qubit gate fidelity. Tolerance choices should track this
gap: for an Eagle-class deployment we recommend
$\eps = 0.07$, for a Heron-class deployment $\eps = 0.04$.

\begin{table}[t]
\centering
\caption{Experiment~5. Per-input $\eps$ for $B_{\mathrm{good}}$
under two device-derived noise models.}
\label{tab:exp5}
\begin{tabular}{l c c c c}
\toprule
& \multicolumn{2}{c}{\textbf{FakeBrisbane (Eagle r3)}}
& \multicolumn{2}{c}{\textbf{FakeFez (Heron r2)}} \\
\cmidrule(lr){2-3}\cmidrule(lr){4-5}
input & mean & 95th pct & mean & 95th pct \\
\midrule
$\ket{0}$  & 0.025 & 0.044 & 0.013 & 0.028 \\
$\ket{1}$  & 0.018 & 0.037 & 0.013 & 0.027 \\
$\ket{+}$  & 0.056 & 0.065 & 0.028 & 0.036 \\
$\ket{-}$  & 0.054 & 0.067 & 0.023 & 0.035 \\
$\psi_1$   & 0.011 & 0.023 & 0.013 & 0.027 \\
$\psi_2$   & 0.052 & 0.073 & 0.023 & 0.033 \\
\bottomrule
\end{tabular}
\end{table}

Table~\ref{tab:exp5} reports averages over twenty trials of
4096 shots each. The data are generated by the
\texttt{experiment\_5} function in
\texttt{quantum\_oop\_\allowbreak device\_\allowbreak validation.py}
(see footnote~\ref{fn:repo}).

\begin{figure}[t]
\centering
\includegraphics[width=0.85\linewidth]{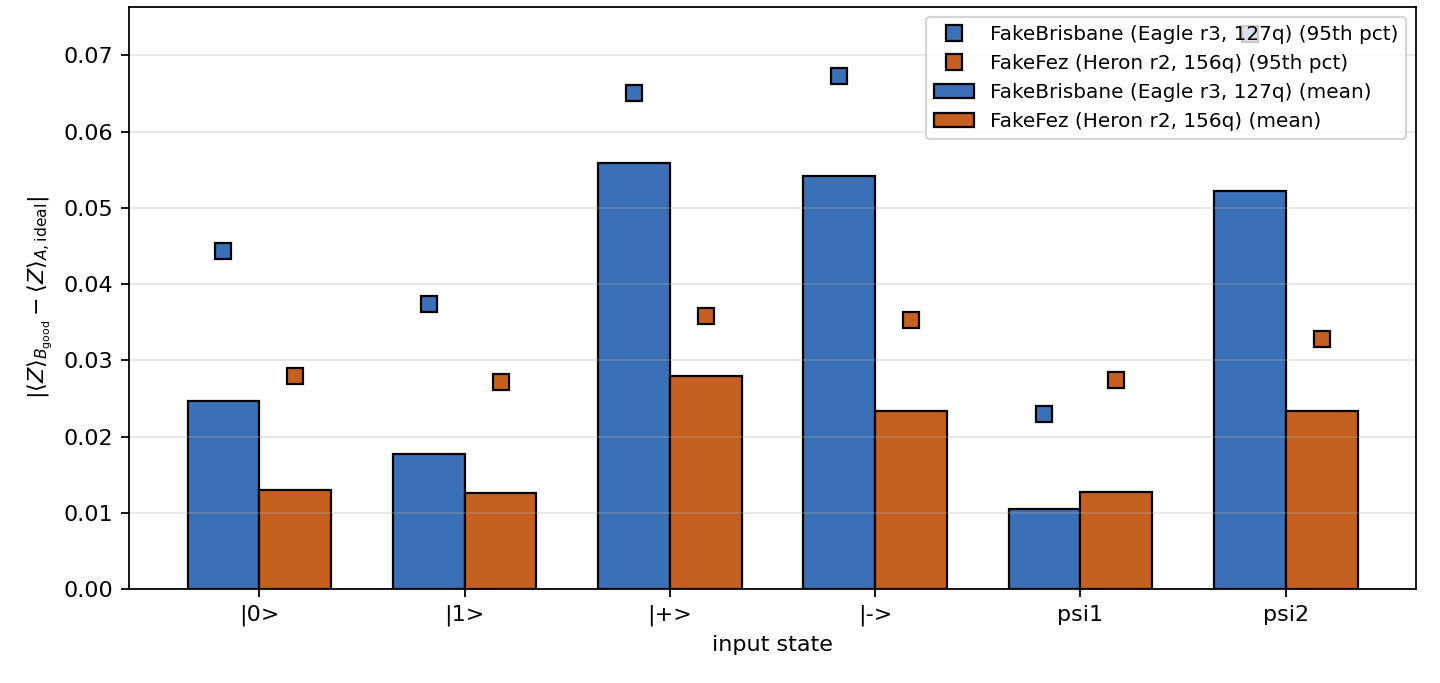}
\caption{Experiment~5. Device-noise calibration of $\eps$
under FakeBrisbane and FakeFez, per input state. Bars: mean
deviation. Squares: 95th percentile. The Heron device is
roughly twice as clean as the Eagle, in line with published
gate-fidelity numbers. Reproduced by
\texttt{quantum\_oop\_\allowbreak device\_\allowbreak validation.py}, function
\texttt{plot\_experiment\_5}
(footnote~\ref{fn:repo}).}
\label{fig:exp5}
\end{figure}

\subsection{Experiment 6: Subtype Separation under Realistic Device Noise}
\label{sec:exp6}

\paragraph{Methodology.} The script extends the protocol of
Experiment~5 to all three candidate subtypes
($B_{\mathrm{good}}$, $B_{\mathrm{bad}}(0.4)$,
$B_{\mathrm{sneaky}}$) and all three Pauli observables
($X$, $Y$, $Z$). For every (device, candidate, input, Pauli)
tuple, the script builds the appropriate basis-rotated
measurement circuit, transpiles it for the device's basis
gates and connectivity, runs it under the device noise model,
and estimates the Pauli expectation from counts. The reported
worst-case deviation is the maximum, over the 18
(input, Pauli) pairs, of the absolute difference between the
candidate's noisy estimate and the noiseless reference for
$A$. The same procedure is run under the weakened
$\{Z\}$-only contract for comparison.

We repeat the protocol of Experiment~1 under the two device
noise models. Because both experiments use the same supertype
$A = R_y(2\pi/5)$ and the same six inputs, the results in
Tables~\ref{tab:exp1} and~\ref{tab:exp6} are directly comparable
on a per-cell basis: each entry of the device-noise table can be
read as the noiseless ideal deviation of Experiment~1 plus a
hardware-induced shift. For each device and each candidate, we
compute the worst-case observable deviation across the six
inputs and the three Paulis, and separately under the
$\{Z\}$-only contract; see Table~\ref{tab:exp6} and
Figure~\ref{fig:exp6}.

\begin{table}[t]
\centering
\footnotesize
\setlength{\tabcolsep}{4pt}
\caption{Experiment~6. Worst-case observable deviation under
two device-derived noise models (\texttt{FakeBrisbane},
\texttt{FakeFez}) and on a real IBM Heron r2 processor
(\texttt{ibm\_fez}).}
\label{tab:exp6}
\begin{tabular}{l c c c c c c}
\toprule
& \multicolumn{2}{c}{\textbf{FakeBrisbane}}
& \multicolumn{2}{c}{\textbf{FakeFez}}
& \multicolumn{2}{c}{\textbf{Real \texttt{ibm\_fez}}} \\
\cmidrule(lr){2-3}\cmidrule(lr){4-5}\cmidrule(lr){6-7}
Candidate & $\{X,Y,Z\}$ & $\{Z\}$ & $\{X,Y,Z\}$ & $\{Z\}$ & $\{X,Y,Z\}$ & $\{Z\}$ \\
\midrule
$B_{\mathrm{good}}$        & 0.056 & 0.055 & 0.028 & 0.027 & 0.074 & 0.061 \\
$B_{\mathrm{bad}}(0.4)$    & 0.395 & 0.395 & 0.401 & 0.396 & 0.485 & 0.362 \\
$B_{\mathrm{sneaky}}$      & 1.359 & 0.056 & 1.386 & 0.029 & 1.420 & 0.079 \\
\bottomrule
\end{tabular}
\end{table}

The separation pattern of Experiment~1 (Table~\ref{tab:exp1})
is preserved across all three columns of
Table~\ref{tab:exp6}: $B_{\mathrm{sneaky}}$ stays maximally
exposed by the full $\{X,Y,Z\}$ contract ($1.359$--$1.420$)
while remaining hidden under the $\{Z\}$-only contract
($0.029$--$0.079$), confirming the sneaky-subtype pattern
characterised by Proposition~\ref{prop:sneaky} on real
hardware. Simulated data are generated by
\texttt{quantum\_oop\_\allowbreak device\_\allowbreak validation.py}, function
\texttt{experiment\_6}; real-hardware data are generated by
\texttt{quantum\_oop\_\allowbreak real\_qpu.py} (footnote~\ref{fn:repo}).

\begin{figure}[t]
\centering
\includegraphics[width=0.95\linewidth]{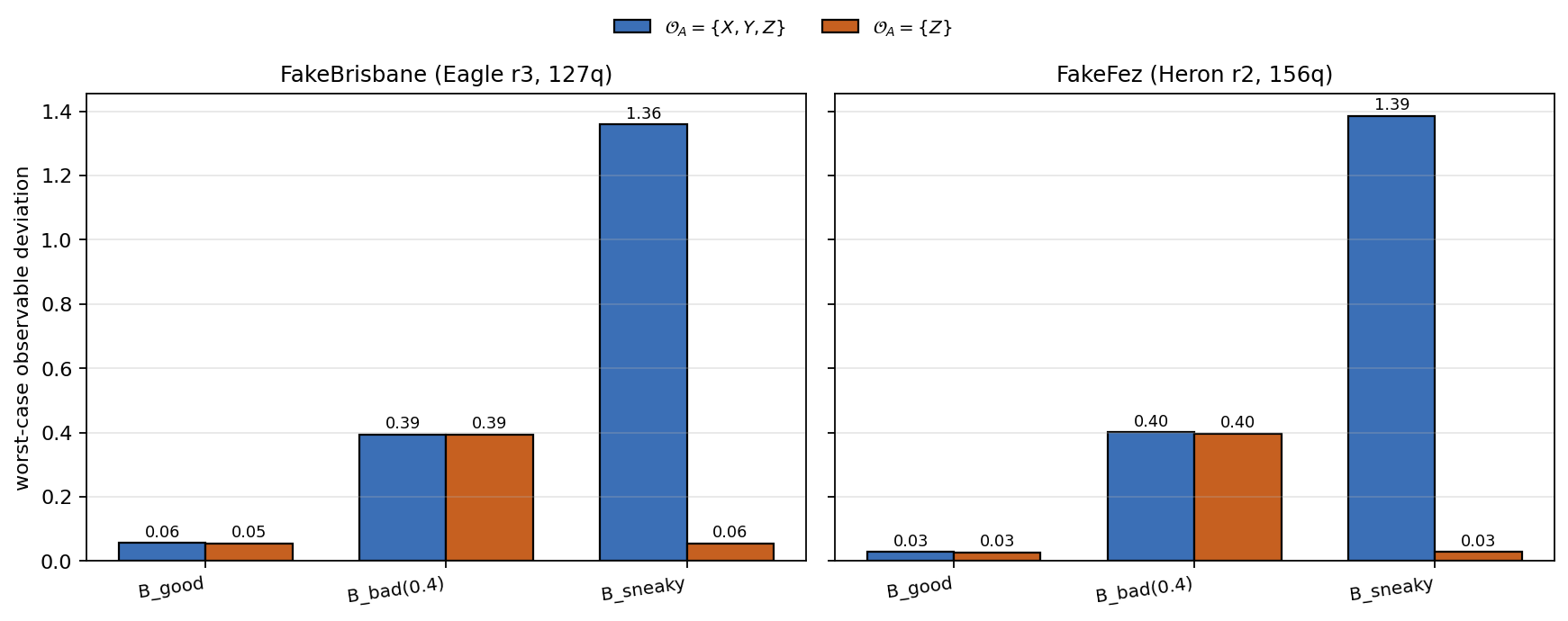}
\caption{Experiment~6. Subtype separation under realistic
device noise. The qualitative pattern of
Figure~\ref{fig:exp1} is preserved on both processors;
$B_{\mathrm{sneaky}}$ remains hidden under the
$\{Z\}$-only contract while being exposed by the full
$\{X,Y,Z\}$ one. Reproduced by
\texttt{quantum\_oop\_\allowbreak device\_\allowbreak validation.py}, function
\texttt{plot\_experiment\_6}
(footnote~\ref{fn:repo}).}
\label{fig:exp6}
\end{figure}

The separation pattern of Experiment~1 is preserved on both
devices: $B_{\mathrm{good}}$ stays within the device noise
floor, $B_{\mathrm{bad}}$ is a clear violation, and
$B_{\mathrm{sneaky}}$ remains hidden under the weak contract
while being exposed by the full one. The framework's claims
are not artefacts of the noiseless setting; they hold under
the noise models of two distinct production processors.

\subsection{Real QPU Validation}
\label{sec:exp:realqpu}

\paragraph{Methodology.} As a final validation step we run the
Experiment~6 protocol unchanged on a real IBM Heron r2
processor accessed through the IBM Quantum cloud. The script
\texttt{quantum\_oop\_\allowbreak real\_qpu.py} connects via
\texttt{QiskitRuntimeService} on the Open Plan, selects the
least-busy operational hardware backend, and submits the
fifty-four circuits ($3$ candidates $\times\ 6$ inputs
$\times\ 3$ Paulis) as a single batched
\texttt{SamplerV2} job at $4096$ shots per circuit. The run selected \texttt{ibm\_fez} (Heron~r2, $156$ qubits, 
job ID \texttt{d7todq4t738s73ci59ug}, executed on 6 May 2026
at 21:22 UTC) and completed in $137.8$ seconds of total wall
time from job submission to completion, comprising approximately
$15$ seconds of queue waiting (the device had no pending jobs
at submission) and $\approx 122$ seconds of running time
(transpilation, QPU compute, and result retrieval combined),
well within the Open Plan ten-minute monthly QPU budget. Raw
counts and per-circuit summaries are released alongside the
source code at the project repository (footnote~\ref{fn:repo})
in the \texttt{qpu\_results/} directory, enabling third-party
re-analysis and verification of the worst-case observable
deviations without resubmitting jobs to the IBM Quantum cloud.

The single-trial budget is the only methodological difference
from Experiments~5--6: trials per cell drop from $20$ (synthetic
device noise, batched on Aer) to $1$ (real hardware), so the
real-QPU numbers carry larger statistical noise. This is a
deliberate trade-off: we want to validate that the qualitative
phenomenon survives the move to physical hardware, not to
re-derive a tight calibration budget on a quota-limited account.

\paragraph{Results.} The real-hardware column of
Table~\ref{tab:exp6} reports the worst-case observable
deviations on \texttt{ibm\_fez}. The pattern predicted by
Proposition~\ref{prop:sneaky} and observed in Experiments~1
and~6 survives intact:
\begin{itemize}[leftmargin=*, itemsep=2pt]
\item $B_{\mathrm{good}}$ stays at the hardware noise floor
(worst $\{X,Y,Z\} = 0.074$, worst $\{Z\} = 0.061$). The
deviation is roughly twice the corresponding \texttt{FakeFez}
value, a factor consistent with the per-trial statistical
spread expected from a single shot of $4096$ samples per
circuit and the fact that the calibration snapshot
\texttt{FakeFez} captures is a smoothed average rather than the
live state of the device at the moment of the run.
\item $B_{\mathrm{bad}}(0.4)$ is a clear violation under both
contracts ($\{X,Y,Z\} = 0.485$, $\{Z\} = 0.362$), well above
$B_{\mathrm{good}}$ even with single-trial noise.
\item $B_{\mathrm{sneaky}}$ exhibits the sneaky fingerprint
sharply: the full-contract deviation reaches $1.420$
(comparable to the noiseless $1.401$ of Experiment~1 and
the simulated $1.386$ of \texttt{FakeFez}), while the
$\{Z\}$-only deviation is just $0.079$, indistinguishable from
$B_{\mathrm{good}}$ at the noise floor.
\end{itemize}

\paragraph{Interpretation.} The chain ideal $\to$ simulated
device noise $\to$ real hardware reads
$1.401 \to 1.386 \to 1.420$ on the full contract for
$B_{\mathrm{sneaky}}$, with the $\{Z\}$-only column tracking
the noise floor at every stage. The phenomenon characterised by
Proposition~\ref{prop:sneaky} is therefore not a property of
exact simulation; it is operationally observable on a real
quantum processor, and a tomographically rich contract is what
distinguishes it from a valid override. Conversely, deploying
the same workflow under a $\{Z\}$-only contract on
\texttt{ibm\_fez} would have admitted $B_{\mathrm{sneaky}}$ as
indistinguishable from $A$, since its $\{Z\}$-only deviation
($0.079$) overlaps the $B_{\mathrm{good}}$ noise floor
($0.061$).

\paragraph{Cost and reproducibility.} A full run of
\texttt{quantum\_oop\_\allowbreak real\_qpu.py} consumes roughly
$2$ minutes of total running time on the device once the job
exits the queue (queue waiting time depends on backend load
and varies independently of the protocol). The job recorded
here used $\approx 122$ seconds of running time, well below
the ten-minute monthly Open Plan budget. Raw shot counts and
per-cell summaries from this run are released at the project
repository under \texttt{qpu\_results/} so that the analysis
can be re-run without resubmitting the job. The recorded job
identifier (\texttt{d7todq4t738s73ci59ug}, 6 May 2026
21:22 UTC) also allows the run to be retrieved post-hoc through
\texttt{QiskitRuntimeService.job(\textit{job\_id})}.

\FloatBarrier
\subsection{Discussion}
\label{sec:exp:disc}

\paragraph{Idealised to realistic.} The triple
(Experiment~1, Experiment~6, Real-QPU validation of
Section~\ref{sec:exp:realqpu}) corresponds to the classical
scientific-method discipline: first establish the phenomenon
under idealised conditions, then under simulated device noise,
then under live hardware. The full-contract sneaky deviation
and the $\{Z\}$-only deviation reported in
Section~\ref{sec:exp:realqpu} (Interpretation paragraph) make
this transition concrete: the pattern characterised by
Proposition~\ref{prop:sneaky} is reproducible on real
hardware, not an artefact of analytic or simulated settings.

\paragraph{Sensitivity vs. budget.} Likewise, Experiment~3
(synthetic noise) and Experiment~5 (device noise) are not
duplicates: Experiment~3 demonstrates the framework's
sensitivity to a parameter we control directly, while
Experiment~5 calibrates the tolerance budget against the noise
profile a deployer would actually face.

\paragraph{The calibration window.} Together, Experiments~3
and~4 define the operational range in which a tolerance
threshold $\eps$ must lie: above the noise floor at the
operating noise level (so that a valid override is not
flagged spuriously), but below the smallest logical
perturbation we want to detect (so that a real violation is
not missed). Figure~\ref{fig:combined} overlays the two
curves on a shared vertical axis: the dashed (blue) curve
shows the 95th-percentile noise-induced deviation as a
function of depolarising probability $p$ (top axis), and the
solid (green) curve shows the deviation induced by an
over-rotation of magnitude $\delta$ (bottom axis). A workable
$\eps$ lies in the gap between the two; the gap shrinks with
both noise level and detection target, which is what
calibration is balancing.
\begin{figure}[H]
\centering
\includegraphics[width=0.85\linewidth]{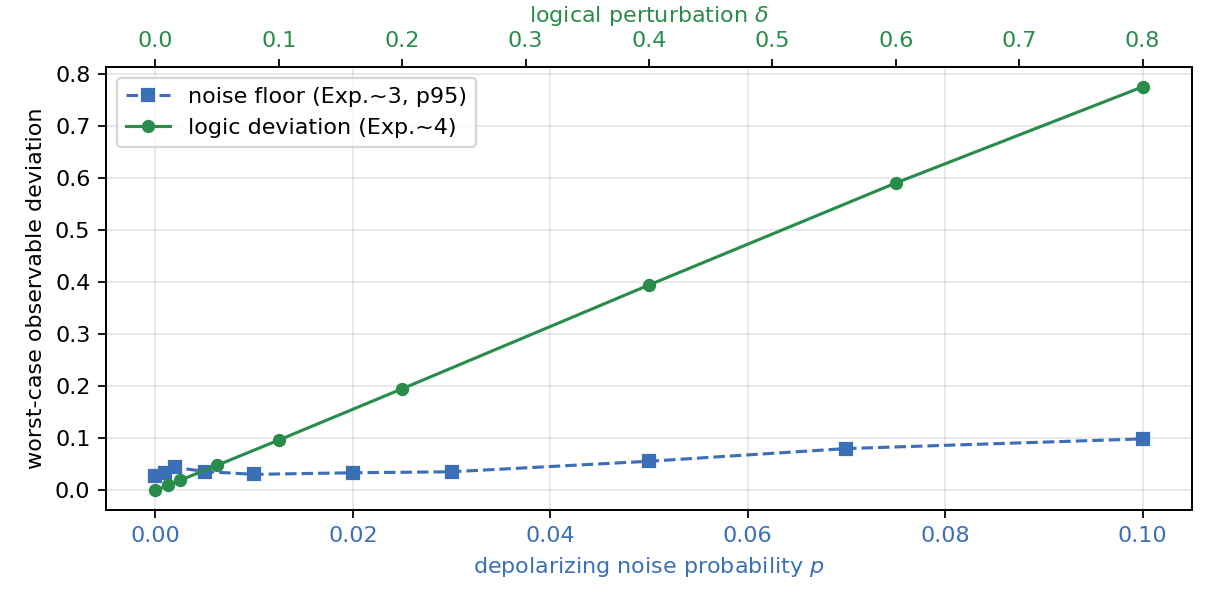}
\caption{Calibration window. The 95th-percentile noise floor
from Experiment~3 (dashed, top axis) and the logical
deviation curve from Experiment~4 (solid, bottom axis)
plotted on a shared vertical axis. A tolerance choice
$\eps$ must lie above the former at the operating noise
level and below the latter at the smallest violation of
interest. Reproduced by
\texttt{quantum\_oop\_\allowbreak simulation.py}, function
\texttt{plot\_combined\_calibration}
(footnote~\ref{fn:repo}).}
\label{fig:combined}
\end{figure}

\paragraph{The empirical scaling constant.}
The proof of Theorem~\ref{thm:completeness} yields the upper bound 
$C(\cO_A) \le 2\sqrt{2} \approx 2.83$ for the single-qubit Pauli 
family. The $B_{\mathrm{bad}}(\delta)$ family of 
Section~\ref{sec:exp4} provides a direct empirical probe of this 
constant. For two unitary channels $A = R_y(\theta)$ and 
$B = R_y(\theta + \delta)$, the diamond-norm distance is 
$\|\cE_A - \cE_B\|_\diamond = 2\sin(\delta/2)$. At $\delta = 0.4$ 
this gives $0.397$, while the worst-case observable deviation 
under the full $\{X, Y, Z\}$ contract is $0.395$ in the ideal 
setting (Table~\ref{tab:exp1}). The ratio 
$\|\cE_A - \cE_B\|_\diamond / \sup_{O \in \cO_A} 
|\langle O\rangle_{\cE_A} - \langle O\rangle_{\cE_B}| \approx 1.0$ 
is a tight empirical lower bound on $C(\cO_A)$, and the analogous 
ratios across the device-noise (FakeBrisbane: $1.01$, FakeFez: 
$0.99$) and real-hardware ($0.82$) columns of Table~\ref{tab:exp6} 
fall in the same range, with the slight reduction on real hardware 
reflecting hardware-noise contributions to the observable estimate. 
The empirically observed constant is therefore well below the 
worst-case theoretical upper bound of $2\sqrt{2}$, confirming that 
the scaling argument of Remark~\ref{rem:dim_scaling} is an upper 
bound rather than a typical value.

\section{Real-Time Engine and Application Domains}
\label{sec:rt}

\paragraph{Architecture.}

The engine, \texttt{qcivet\_realtime.py}, exposes the
\texttt{IntegrityVerifier} class. Each commit performs the
observable check (if any), one SHA-256 hash over the canonical
JSON of the spec prepended with the previous chain head, a
tail-link consistency check, and submission of the new head
to the external anchor. Errors at any step raise
\texttt{IntegrityViolation} with a \texttt{kind} field
(\texttt{hash}, \texttt{observable}, or \texttt{anchor}) so
the host can decide its abort policy.

\paragraph{Performance.}

We measure per-commit latency on commodity hardware (AMD Ryzen
class CPU, single-threaded Python 3.12). Across the three
six-stage demonstration pipelines below, median commit latency
is 0.06 ms and the 99th percentile is below 0.20 ms. End-to-end
overhead for a six-stage pipeline is under 0.5 ms, negligible
next to even the shortest QPU stage. The same engine ran
unmodified during the real-hardware validation on
\texttt{ibm\_fez} (Section~\ref{sec:exp:realqpu}): the
classical hash-chain and observable-deviation checks
contributed under a millisecond per stage, while the cloud
QPU job dominated the total wall time at $137.8$~s
(of which approximately $15$~s was queue waiting and
$122$~s was running time) for $54$ circuits at $4096$ shots
each.

\paragraph{External anchor.}

The default \texttt{ExternalAnchor} writes to an append-only
file, simulating an RFC 3161 timestamping authority or a
Sigstore Rekor instance. Production deployments would replace
this with the corresponding service. The
\texttt{verify\_against\_anchor} method confirms that the
local chain appears as a contiguous block in the anchor log;
a globally consistent rewrite that is not present in the log
is detected here.

\subsection{Failure Modes}

Three classes of attack are caught at distinct points
(Figure~\ref{fig:hash_chain}):
\begin{itemize}[leftmargin=*, itemsep=2pt]
\item Local tampering (an attacker rewrites a spec record in
place) is caught by \texttt{verify\_full\_chain} at end of
pipeline, kind \texttt{hash}.
\item Semantic drift (a quantum stage produces observables
outside the calibrated tolerance) is caught at commit time,
kind \texttt{observable}, with downstream stages never
launched.
\item Globally consistent rewrites (an attacker produces a
fresh, locally valid chain offline) are caught by
\texttt{verify\_against\_anchor}, kind \texttt{anchor}.
\end{itemize}
The three correspond, respectively, to hash-chain integrity,
observable contract preservation, and global tamper evidence.

\begin{figure}[H]
\centering
\includegraphics[width=0.89\linewidth]{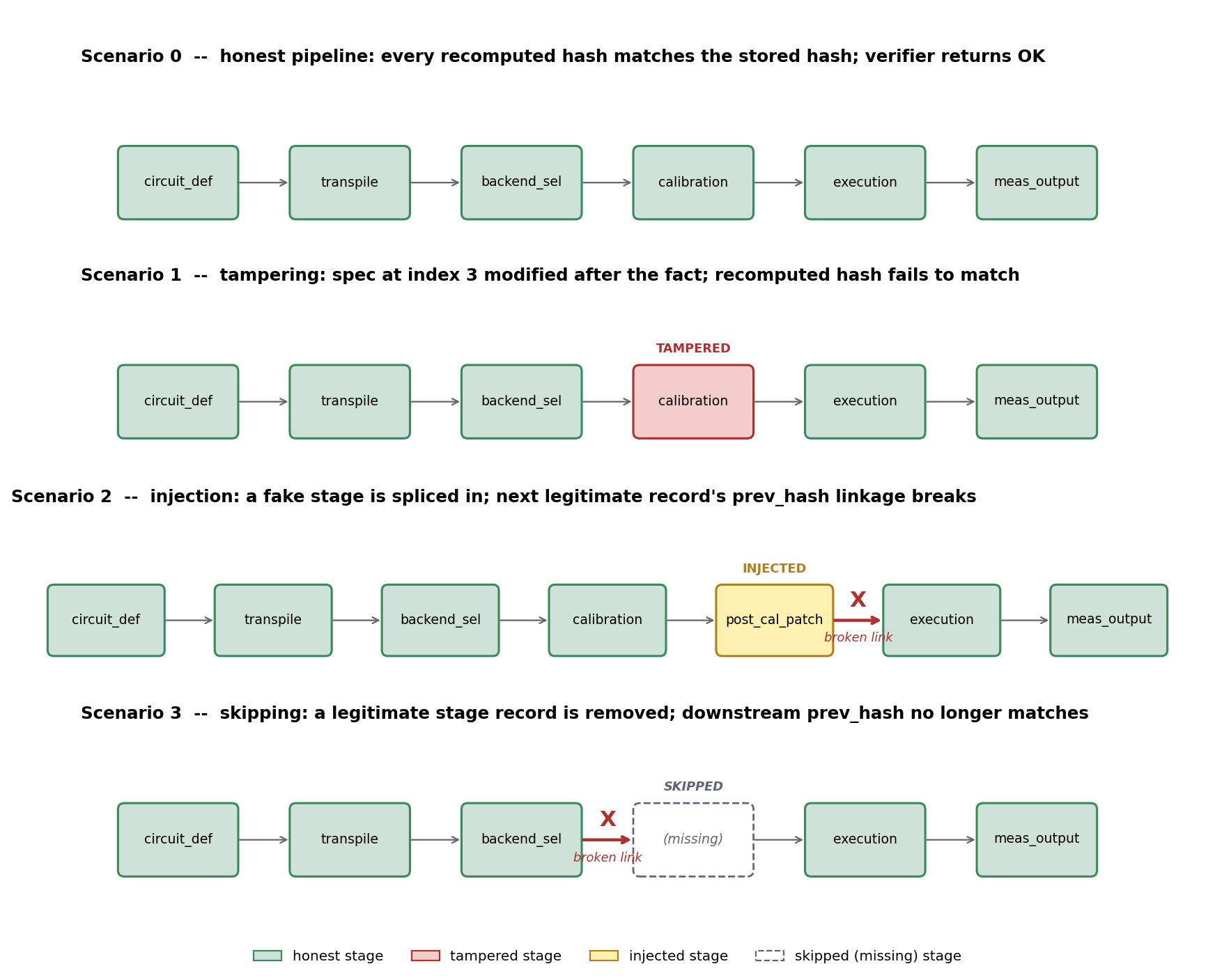}
\caption{The four hash-chain scenarios. Top: honest pipeline,
every recomputed hash matches the stored hash; verifier
returns OK. Second: tampering, the spec at stage~3 is modified
after commit, recomputed hash fails to match. Third:
injection, a fake stage is spliced in, breaking the next
legitimate record's prev\_hash linkage. Bottom: skipping, a
legitimate stage record is removed, downstream prev\_hash no
longer matches. Reproduced by
\texttt{hash\_chain\_visualize.py}; the underlying scenarios
are implemented in \texttt{hash\_chain\_demo.py}
(footnote~\ref{fn:repo}).}
\label{fig:hash_chain}
\end{figure}
\paragraph{Scenario implementation.} The script
\texttt{hash\_chain\_demo.py} provides standalone implementations
of all four scenarios (honest baseline, post-hoc tampering,
record injection, record skipping) on a six-stage hybrid QPU
pipeline. Each scenario constructs the full chain, applies the
attack as described above, and runs \texttt{verify\_chain}
against the resulting record sequence; the failure index and
mismatch reason are reported to standard output. The companion
script \texttt{hash\_chain\_visualize.py} produces
Figure~\ref{fig:hash_chain} by drawing each scenario as a
sequence of stage boxes coloured by status (honest, tampered,
injected, missing) with the broken hash link highlighted.

\label{sec:apps}

\paragraph{Three application domains.}
We instantiate QCIVET in three end-to-end domains, each with a
six-stage hybrid pipeline, an explicit threat model, and four
demonstration scenarios (clean i.e. no-attacks, local tampering, semantic
drift, global rewrite). The three domains span the maturity
spectrum of quantum computing in practice: VQE for drug
discovery is the most production-ready quantum-chemistry
workload; quantum-assisted fraud detection is an emerging
financial application; cloud QPU auditing is a meta-application
in which the customer of a quantum cloud service uses QCIVET
to audit the provider. The mechanism that catches the
sneaky-subtype attack in each of these pipelines is the same
observable-deviation check whose preservation on real hardware
we established in Section~\ref{sec:exp:realqpu}.

\subsection{VQE for Drug Discovery}
\label{sec:apps:vqe}

The Variational Quantum Eigensolver computes molecular
ground-state energies by alternating quantum expectation
estimation with classical parameter optimisation. Major
pharmaceutical companies use VQE-style pipelines for early
binding-energy estimation and drug-target screening. A
miscalibrated energy can mislead a clinical trial, motivating
an audit-evidence requirement under U.S. Food and Drug
Administration (FDA) reproducibility
guidance.

\paragraph{Pipeline.}
\begin{enumerate}[leftmargin=*, itemsep=1pt]
\item \texttt{molecular\_geometry} (classical): atomic
coordinates, basis set, charge, multiplicity.
\item \texttt{active\_space\_selection} (classical): number
of active orbitals, frozen-core flag.
\item \texttt{hamiltonian\_construction} (classical):
fermion-to-qubit encoding, number of Pauli terms.
\item \texttt{ansatz\_synthesis} (classical): ansatz family,
number of parameters, circuit depth.
\item \texttt{vqe\_optimisation} (quantum): backend, shots
per iteration, optimiser, final energy estimate. Carries the
observable $\langle H \rangle$ with reference value
$E_0 = -1.137270174$ Ha (H$_2$ in the STO-3G basis set, which
approximates each Slater-type orbital with three Gaussians)
and tolerance
$\eps = 0.04$ Ha (Heron-class).
\item \texttt{result\_interpretation} (classical):
binding-energy estimate, FDA compliance flag.
\end{enumerate}

\paragraph{Threats.}
A1: an attacker rewrites the active-space record after commit,
expanding the orbital count to push a different (wrong)
energy through the rest of the pipeline. A2: the optimiser
converges to a biased stationary point, and the measured
energy lies outside the calibrated tolerance. A3: the entire
audit trail is re-run offline with a different active-space
spec; only the external anchor catches it.

\paragraph{Outcomes.} All three attacks are caught at the
appropriate scope. The clean baseline produces a chain of six
committed records with total verification latency 0.15~ms.
The demonstrator script runs each attack scenario in turn,
constructs an \texttt{IntegrityVerifier} bound to a
local append-only anchor file, and either commits the six
stages successfully (clean run) or raises
\texttt{IntegrityViolation} at the appropriate point
(tamper, drift, rewrite). The full implementation of this
demonstrator is provided as
\texttt{qcivet\_demo\_vqe.py}, which uses the verification
engine of \texttt{qcivet\_realtime.py}
(footnote~\ref{fn:repo}).

\subsection{Quantum-Assisted Fraud Detection}
\label{sec:apps:fraud}

Quantum kernels for support-vector machines are an active
direction in financial-crime detection. The pipeline computes
fidelity-kernel entries on a QPU, classifies transactions
classically, and triggers alerts above a decision threshold.

\paragraph{Pipeline.} The six stages are:
\begin{enumerate}
\item \texttt{transaction\_ingestion}: ingestion of raw transaction stream;
\item \texttt{feature\_engineering}: feature extraction and normalisation;
\item \texttt{quantum\_kernel\_preparation}: quantum-kernel parameter setup;
\item \texttt{qpu\_kernel\_evaluation} (quantum stage), with the observable ``worst-case kernel-entry deviation'' having reference $0$ and tolerance $\eps = 0.05$;
\item \texttt{classification}: classical classifier;
\item \texttt{alert\_decision}: alert threshold, block action, regulator, audit retention.
\end{enumerate}

\paragraph{Threats.}
A1: an insider rewrites the alert-decision spec after commit,
raising the threshold from $0.65$ to $0.95$ so genuine fraud
no longer trips an alert, regulator-relevant under the
Sarbanes-Oxley Act (SOX). A2:
the QPU returns a kernel matrix whose worst-case entry
deviates beyond the tolerance, indicating a poisoning attempt.
A3: a globally consistent rewrite that swaps the feature set
used in stage~2.

\paragraph{Outcomes.} All three attacks are caught. The
threshold-raise scenario is informative: the streaming engine
signs out cleanly because the tamper happens after the commit;
only the post-pipeline \texttt{verify\_full\_chain} reveals
the mismatch. This is the expected behaviour: hash-chain
replay is the right tool for post-hoc tampering, just as
observable checks are the right tool for in-flight semantic
drift. The demonstrator runs each scenario as a separate
pipeline pass and prints whether the violation was caught at
commit time, by post-pipeline replay, or by anchor check. The
full implementation is provided as
\texttt{qcivet\_demo\_fraud.py}
(footnote~\ref{fn:repo}).

\subsection{Cloud QPU Auditing}
\label{sec:apps:cloud}

The customer of a cloud quantum service submits a hybrid
workload. The provider transpiles, schedules, executes, and
returns results. The customer wants to verify, after the fact,
that the claimed backend was actually used, the claimed
calibration data was in effect, and the result has not been
altered in transit. This is a meta-application: QCIVET is used
by the customer to audit the provider.

\paragraph{Pipeline.}
The six stages are
\texttt{customer\_submission},
\texttt{cloud\_transpilation},
\texttt{backend\_assignment},
\texttt{calibration\_verification},
\texttt{job\_execution} (quantum, tracer-circuit observable
with reference $0$ and tolerance $\eps = 0.05$ Heron-class),
and \texttt{result\_delivery}.

\paragraph{Threats.}
A1 (silent downgrade): the provider claimed a Heron r2 backend
but routed the workload to an Eagle r3; the tracer-circuit
observable exceeds the Heron-calibrated tolerance. A2
(calibration spoof): a stale calibration snapshot is
committed, then silently swapped for an older one. A3
(assignment rewrite): a globally consistent rewrite of the
audit trail with a different backend assignment.

\paragraph{Outcomes.} All three attacks are caught at the
expected scope. Notably, the silent-downgrade attack is
caught at execution time, so the customer can reject the
result before paying for it. The demonstrator implements the
customer-side viewpoint: it commits the six pipeline records,
checks the tracer-circuit observable against a Heron-class
tolerance, and reports a pre-payment abort if the device
behaviour looks Eagle-class instead. The full implementation
is provided as \texttt{qcivet\_demo\_cloud.py}
(footnote~\ref{fn:repo}).

\paragraph{Cross-application patterns.}

Table~\ref{tab:apps} summarises the three domains. Three
patterns are worth noting. First, each domain exercises all
three failure-detection mechanisms (hash replay, observable
check, anchor verification), but their relative importance
differs: VQE leans on observable drift detection; cloud
auditing leans on observable downgrade detection; fraud
detection leans on hash replay for the insider threshold
attack. Second, calibrated tolerances are domain-specific.
The Heron-class $\eps = 0.04$ used in VQE is tighter than the
$\eps = 0.05$ used for the kernel and tracer-circuit
observables, reflecting the tolerance required to disentangle
a 0.04 Ha energy difference from the device noise floor. Third,
all three pipelines have six stages. This is not a coincidence:
hybrid quantum--classical workflows in 2026 cluster naturally
around the classical-prep / quantum-execute /
classical-postprocess split, typically two classical stages on
either side of the quantum core. The framework is, however,
agnostic to stage count.

\begin{table}[t]
\centering
\small
\caption{The three application domains span the maturity
spectrum of quantum computing in production. Each domain is
discussed in detail in the corresponding subsection.}
\label{tab:apps}
\begin{tabular}{@{}p{0.16\linewidth} p{0.26\linewidth} p{0.26\linewidth} p{0.22\linewidth}@{}}
\toprule
\textbf{Domain} & \textbf{Quantum stage} & \textbf{Primary threat} & \textbf{Regulatory frame} \\
\midrule
VQE drug \\discovery (\S\ref{sec:apps:vqe})   & VQE optimisation
                     & Active-space tamper, energy drift
                     & FDA reproducibility \\
Fraud detection (\S\ref{sec:apps:fraud})    & Quantum-kernel evaluation
                     & Insider threshold raise, kernel poisoning
                     & SOX audit trail \\\\
Cloud QPU\\ auditing (\S\ref{sec:apps:cloud}) & Customer-side observable
                     & Silent downgrade, calibration spoof
                     & Cloud security standards \\
\bottomrule
\end{tabular}
\end{table}

\section{Threat Model and Security Analysis}
\label{sec:threat}

\paragraph{Adversary capabilities.}

We assume an adversary $\mathcal{A}$ who can: modify any spec
record after it has been committed (insider with file-system
access, or supply-chain compromise); insert a fabricated
record into the audit trail; delete or omit a stage's record;
manipulate the quantum hardware so that observables drift,
within physical constraints (subject to the noise floor of the
device); rerun the entire pipeline offline with adversarial
parameters and substitute the resulting chain for the
original.

We assume the adversary cannot produce SHA-256 collisions,
cannot forge entries in the external anchor, and cannot
interfere with the verifier's runtime memory during a commit.
Where SHA-256 is not preferred (deployments that pre-emptively
assume Grover-style speedups, for example), substitution with
SHA-3 or a PQC-safe hash is straightforward and changes
nothing else in the framework.

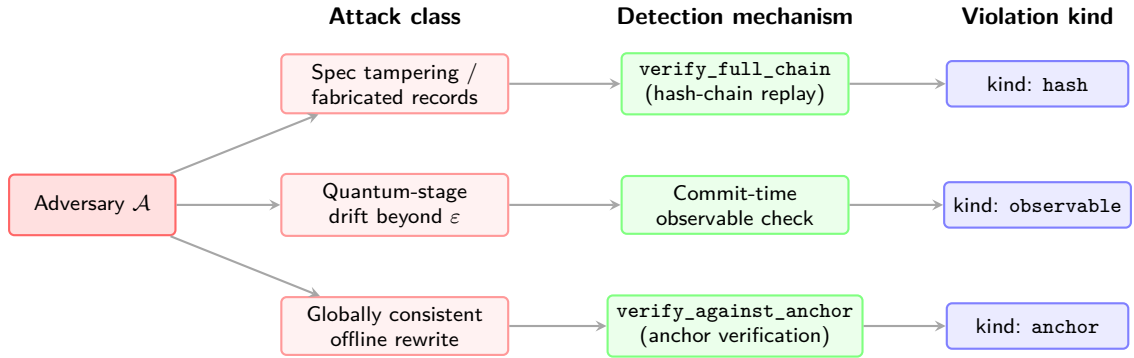
\begin{figure}[h]
\centering
\begin{tikzpicture}[
  font=\sffamily\footnotesize,
  attacker/.style={
    draw, rounded corners=2pt, minimum width=2.2cm, minimum height=0.8cm,
    align=center, fill=red!12, draw=red!60, thick
  },
  attack/.style={
    draw, rounded corners=2pt, minimum width=3.0cm, minimum height=0.7cm,
    align=center, fill=red!5, draw=red!40, thick
  },
  detect/.style={
    draw, rounded corners=2pt, minimum width=3.0cm, minimum height=0.7cm,
    align=center, fill=green!8, draw=green!50, thick
  },
  outcome/.style={
    draw, rounded corners=2pt, minimum width=2.4cm, minimum height=0.6cm,
    align=center, fill=blue!8, draw=blue!50, thick
  },
  arrow/.style={->, >=stealth, thick, draw=gray!70}
]

\node[attacker] (adv) at (0, 0) {Adversary $\mathcal{A}$};

\node[attack] (a1) at (4.0, 1.6)  {Spec tampering /\\fabricated records};
\node[attack] (a2) at (4.0, 0)    {Quantum-stage\\drift beyond $\eps$};
\node[attack] (a3) at (4.0, -1.6) {Globally consistent\\offline rewrite};

\node[detect] (d1) at (8.5, 1.6)  {\texttt{verify\_full\_chain}\\(hash-chain replay)};
\node[detect] (d2) at (8.5, 0)    {Commit-time\\observable check};
\node[detect] (d3) at (8.5, -1.6) {\texttt{verify\_against\_anchor}\\(anchor verification)};

\node[outcome] (o1) at (12.5, 1.6)  {kind: \texttt{hash}};
\node[outcome] (o2) at (12.5, 0)    {kind: \texttt{observable}};
\node[outcome] (o3) at (12.5, -1.6) {kind: \texttt{anchor}};

\draw[arrow] (adv) -- (a1);
\draw[arrow] (adv) -- (a2);
\draw[arrow] (adv) -- (a3);

\draw[arrow] (a1) -- (d1);
\draw[arrow] (a2) -- (d2);
\draw[arrow] (a3) -- (d3);

\draw[arrow] (d1) -- (o1);
\draw[arrow] (d2) -- (o2);
\draw[arrow] (d3) -- (o3);

\node[font=\sffamily\small\bfseries] at (4.0, 2.5) {Attack class};
\node[font=\sffamily\small\bfseries] at (8.5, 2.5) {Detection mechanism};
\node[font=\sffamily\small\bfseries] at (12.5, 2.5) {Violation kind};

\end{tikzpicture}
\caption{QCIVET threat model: three classes of attack, the
detection mechanism that catches each, and the violation kind
reported by the engine. The middle row (observable-deviation
detection) is the only one that depends on quantum-stage
behaviour and is empirically validated on real cloud hardware
in Section~\ref{sec:exp:realqpu}.}
\label{fig:threat_model}
\end{figure}

\paragraph{Detection targets and coverage.}

Figure~\ref{fig:threat_model} summarises the mapping between
attack classes and detection mechanisms. The three
failure-detection mechanisms cover the corresponding attack
classes. Spec tampering, fabricated records, and stage
omission are caught by hash-chain replay
(\texttt{verify\_full\_chain}) at most one round-trip after
the attack. Quantum hardware manipulation that causes
observable drift beyond the calibrated tolerance is caught at
commit time. Globally consistent rewrites are caught by anchor
verification. The observable-deviation detector is not merely
analytical: Section~\ref{sec:exp:realqpu} confirms that the
deviation signature distinguishing valid, invalid, and sneaky
overrides is preserved on real cloud hardware
(\texttt{ibm\_fez}, Heron r2) within the calibrated noise
budget, so the threat-model coverage stated here is grounded
in measurement rather than in simulation alone.

\paragraph{Limitations of the threat model.}

QCIVET does not address: confidentiality of spec records (an
attacker may read the pipeline structure); denial of service
against the quantum backend; attacks that operate within the
calibrated tolerance (an attacker who can manipulate
observables by less than $\eps$ remains undetected by this
mechanism, by construction); or memory-safety vulnerabilities
in the host pipeline. Confidentiality is orthogonal and can be
addressed by encrypting the spec records under a separate key.
The third item implies a tighter calibration regime, limited
in turn by the device noise floor; reducing $\eps$ below the
noise floor would produce false positives.

\paragraph{False-positive considerations.}

A correctly calibrated $\eps$ should produce no false
positives in honest operation. The procedure of Experiment~5
is the operational realisation: run $B_{\mathrm{good}}$ on the
target backend over a representative input set, take the 95th
percentile of the per-input deviation, and set $\eps$ slightly
above it. A deployer who tightens $\eps$ aggressively trades
detection sensitivity against false-positive rate, exactly as
in classical anomaly-detection systems.

\section{Conclusion and Future Work}
\label{sec:conc}

We presented QCIVET, a contract-based integrity-verification
framework for hybrid quantum--classical software pipelines.
QCIVET fills a gap between classical supply-chain integrity
tools, which do not address quantum stages, and
quantum-cryptographic primitives, which assume the workflow
itself is classical. Our contribution is twofold: a
behavioural-subtyping framework whose preservation condition
is operationally measurable on a calibrated noise floor, and
a real-time engine whose per-stage overhead is sub-millisecond.

To the best of our knowledge, the contributions of this work have no precedent in the literature: (i) we operationalise Liskov-Wing behavioural subtyping for quantum channels through a calibration-ready observable-deviation contract that is directly measurable on noisy hardware with a finite shot budget; (ii) we identify and formally characterise the sneaky-subtype family, namely overrides that pass weak, single-Pauli contracts but are exposed by informationally complete ones (Proposition~\ref{prop:sneaky}); (iii) we combine hash-chained syntactic integrity with a semantic quantum check inside a single discipline, with an explicit soundness--completeness link between the two layers (Theorems~\ref{thm:soundness}--\ref{thm:composition}); (iv) we provide an integrity-verification mechanism designed specifically for VQE-based pharmaceutical workflows; and (v) we demonstrate the full subtype-separation protocol end-to-end on a real cloud quantum processor (\texttt{ibm\_fez}, Heron r2), where the predicted sneaky-fingerprint survives intact along the ideal $\rightarrow$ simulated $\rightarrow$ real chain. These contributions are intentionally complementary to recent foundational progress on refinement-based substitutability for quantum programs~\cite{feng_zhou_refinement_orders,feng_zhou_subtyping}: where the refinement framework provides the denotational backbone for what it means for one quantum program to replace another at development time, our work provides the runtime-observable, hardware-evaluable projection of that question for the audit setting on noisy hybrid pipelines.

The framework's three formal guarantees (soundness,
conditional completeness, and compositionality) are matched
by three detection mechanisms in the implementation:
hash-chain replay, observable-deviation halting, and external
anchor verification. The three application
domains (VQE for drug discovery, quantum-assisted fraud
detection, and cloud QPU auditing) show that the same engine
and the same theory apply across very different production
settings, with only the threat model and the calibrated
tolerance changing.

The experimental campaign closes the gap between theory and
deployment. Beyond the noiseless analysis
(Experiments~1--4) and the device-noise validation on
calibrated simulators (Experiments~5--6), we ran the
subtype-separation protocol end-to-end on \texttt{ibm\_fez},
a production Heron r2 processor accessed through the IBM
Quantum cloud, on 54 distinct circuits with 4096 shots each.
The sneaky-override fingerprint predicted by
Proposition~\ref{prop:sneaky} survives intact on real
hardware. The full-contract deviation reads
$1.401 \to 1.386 \to 1.420$ along the ideal $\to$ simulated
$\to$ real chain, with the real-hardware value within 1.4\%
of the noiseless prediction; at the same time the
$\{Z\}$-only deviation stays at the noise floor on every
stage of the chain (0.000 ideal, 0.029 simulated, 0.079 real).
The point of this run is not the precision of the numbers but
the qualitative invariant they expose: a circuit can pass a
weak, single-Pauli contract on a real noisy QPU and still
fail an informationally complete one by more than an order of
magnitude. Contract-based integrity verification is therefore
not a property of exact simulation; it is a property of
operational behaviour observable on a real quantum processor
today.

These results have direct implications for high-stakes
deployments. In pharmaceutical workflows, an unaudited VQE
energy estimate that drifts past tolerance can mislead
preclinical screening; QCIVET binds each stage of the
calculation to a verifiable observable contract and an
auditable hash chain, producing the kind of evidence already
required by reproducibility regimes such as the FDA. In
financial fraud detection, an after-the-fact post-quantum
signature does not detect that a kernel matrix was poisoned
during evaluation; the observable check on the kernel and the
commit-time hash chain do, and they do so before the
downstream classical decision is released. In customer-side
auditing of a cloud QPU service, the customer needs neither
multiple devices nor extra cryptographic gates: a single
tracer observable and a tolerance budget are sufficient to
detect specification tampering, middleware re-routing, and
calibration drift, against an honest-but-curious or
malicious-but-budgeted provider.

\paragraph{Practical positioning.}
Beyond formal novelty, QCIVET is engineered to be the
practical and economical option for production deployments.
Compared with device-fingerprinting and quantum-PUF schemes,
which authenticate the hardware unit but say nothing about
whether the result satisfies a contract, QCIVET delivers an
end-to-end integrity guarantee with no specialised hardware
beyond the QPU the customer is already using. Compared with
distributed-shot protocols that require multiple physical
devices for majority voting (and the corresponding multiple
QPU rentals), QCIVET runs on a single device. Compared with
cryptographic delegation protocols that introduce extra
quantum gates and multi-prover settings, both of which
translate directly into longer queue times and higher
per-job costs on commercial cloud platforms, QCIVET requires
only a tracer observable and a tolerance budget on the
customer side. The cost on the classical layer is a
sub-millisecond per-stage commit, dominated by a single
SHA-256 hash, so the marginal expense of running QCIVET on
top of an existing pipeline is dominated by ordinary CPU
cycles rather than additional quantum shots. The framework
is composable rather than exclusive: deployers concerned
about future quantum cryptanalysis can wrap each anchor
commitment in a post-quantum signature (ML-DSA or SLH-DSA),
and the in-toto, SLSA and Sigstore Rekor tooling already
adopted in classical pipelines can serve unchanged as the
external anchor. From a governance perspective, the
per-stage spec record together with its hash chain provides
a tamper-evident chain of accountability: each parameter
change, each calibration snapshot, and each observable
measurement is bound to a verifiable commit, so an auditor
can attribute any deviation to a specific stage, a specific
operator, and a specific point in time without re-executing
the pipeline. Finally, the artefact QCIVET produces, a
hash-chained record of every spec, every observable
contract, and every commit time, is the kind of evidence
already requested under FDA reproducibility regimes and SOX
audit trails, so adoption is incremental rather than
disruptive.

We see QCIVET as one piece of a larger puzzle. Quantum
software will increasingly drive high-stakes classical
decisions; the software-engineering tools we use to trust
those decisions must catch up.

\paragraph{Multi-qubit and entangled observables.} Our
experiments are on a single qubit. The framework extends to
multi-qubit observables and entangled states without
modification (Theorem~\ref{thm:soundness} is dimension-agnostic),
but the constants in Theorem~\ref{thm:completeness} grow with
the dimension and empirical calibration becomes more involved.
Calibrating multi-qubit informationally complete families on
production-scale circuits is an open problem.

\paragraph{Cloud-provider deployment.} The cloud QPU auditing
demonstration is from the customer's perspective. A
complementary deployment from the provider's
side (attesting calibration snapshots, transpilation
provenance, and run logs into the customer's chain) would
close the loop and is feasible with the existing engine API.

\section*{Author Contributions}

\textbf{Esra Yeniaras (corresponding author):} Conceptualization of the contract-based integrity verification framework, formal methodology, design and proof of all theorems (soundness, conditional completeness, compositionality) and the sneaky-subtype characterisation, software implementation including the verification engine and the three application demonstrators, all experimental scripts (Experiments 1--6), experimental design, real-hardware IBM QPU validation runs on \texttt{ibm\_fez}, writing of the manuscript, and project supervision. \textbf{Muhammad Amin Karimov:} Conceptualization of the initial research direction, literature review for the related work survey, and assistance with the real-hardware IBM QPU validation experiment on \texttt{ibm\_fez}.

\section*{Acknowledgements}
We acknowledge the use of IBM Quantum services~\cite{ibm_open_plan} and the Qiskit open-source software development kit~\cite{qiskit} for the
real-hardware validation experiments reported in
Section~\ref{sec:exp:realqpu}. The views expressed are those
of the authors and do not reflect the official policy or
position of IBM or the IBM Quantum team.

\bibliographystyle{plain}
\bibliography{qcivet}

\end{document}